\documentclass[a4paper]{article}
\usepackage[pagebackref,colorlinks=true, urlcolor=blue,linkcolor=blue,citecolor=blue,pdfstartview=FitH]{hyperref}

\usepackage{amsmath,amsfonts,amssymb,amsthm}
\usepackage{thmtools, thm-restate}

\usepackage[linesnumbered,ruled, lined, commentsnumbered, longend]{algorithm2e}

\usepackage{fullpage,color}
\usepackage{setspace}
\usepackage{mathpazo}

\usepackage[nameinlink,capitalize]{cleveref}
\renewcommand{\eqref}[1]{\hyperref[#1]{(\ref*{#1})}}

\usepackage[colorinlistoftodos]{todonotes}
\usepackage{nicefrac}

\newcommand{\yfnote}[1]{\todo[color=red!20!blue!15, size=\footnotesize]{yf: #1}}

\newtheorem{theorem}{Theorem}[section]
\newtheorem{lemma}[theorem]{Lemma}
\newtheorem{corollary}[theorem]{Corollary}
\newtheorem{claim}[theorem]{Claim}

\newtheorem{definition}[theorem]{Definition}
\newtheorem{remark}[theorem]{Remark}

\providecommand{\hit}[1]{{\rm Hit}_{#1}}

\renewcommand{\epsilon}{\varepsilon}
\renewcommand{\phi}{\varphi}
\newcommand{\eps}{\epsilon}

\newcommand{\calF}{{\mathcal{F}}}
\newcommand{\F}{{\mathcal{F}}}
\newcommand{\fdom}[1]{\colon #1 \to \Sigma}

\providecommand{\paramq}{\alpha}
\providecommand{\bits}{\{0,1\}}

\DeclareMathOperator{\bdeg}{bdeg}
\DeclareMathOperator{\rej}{rej}

\providecommand{\rTpd}{\rej_{T_{p,d}}}
\providecommand{\rTd}{\rej_{T_d}}

\newcommand{\restr}{\!\!\upharpoonright}
\newcommand{\bi}[2]{\binom{#1}{#2}}
\newcommand{\bind}[1]{\bi{#1}{d}}
\newcommand{\binld}[1]{\bi{#1}{\leq d}}
\newcommand{\bnd}{\bind{[n]}}

\DeclareMathOperator*{\EE}{\mathbb{E}}
\DeclareMathOperator{\agree}{agree}
\newcommand{\lagree}[1]{\stackrel{#1}{\approx}}
\newcommand{\Prob}[2][]{\Pr_{{#1}}\left[#2\right]} 

\newcommand{\set}[1]{\{#1\}}

\newcommand{\etal}{{\em et al.}}

\title{Agreement tests on graphs and hypergraphs \footnote{A preliminary version of this paper and the parallel work~\cite{DinurFH-ks} appeared in  \emph{Proc.\ $30$th Annual {ACM}-{SIAM} Symp.\ on Discrete
  Algorithms (SODA)}~\cite{DinurFH2019}.}}
\author{Irit Dinur\thanks{Weizmann Institute of Science, ISRAEL. email: {\tt irit.dinur@weizmann.ac.il}. Research supported by ERC-CoG grant 772839.}
\and Yuval Filmus\thanks{Technion — Israel Institute of Technology, ISRAEL. email: {\tt yuvalfi@cs.technion.ac.il}. Taub Fellow --- supported by the Taub Foundations. The research was funded by ISF grant 1337/16.}
 \and Prahladh Harsha\thanks{Tata Institute of Fundamental Research, INDIA. email: {\tt prahladh@tifr.res.in}. Research supported by the Department of Atomic Energy, Government of India, under project no. 12-R\&D-TFR-5.01-0500 and in part by the UGC-ISF grant and Swarnajayanti fellowship. Part of this work was done when the author was visiting the Weizmann Institute of Science.}}

\begin{document}

\begin{titlepage}

\pagenumbering{gobble}

\maketitle
\begin{abstract}
  Agreement tests are a generalization of low degree tests that capture a local-to-global phenomenon, which forms the combinatorial backbone of most PCP constructions. In an agreement test, a function is given by an ensemble of local restrictions. The agreement test checks that the restrictions agree when they overlap, and the main question is whether average agreement of the local pieces implies that there exists a global function that agrees with most local restrictions.

There are very few structures that support agreement tests, essentially either coming from algebraic low degree tests or from direct product tests (and recently also from high-dimensional expanders). In this work, we prove a new agreement theorem which extends direct product tests to higher dimensions, analogous to how low degree tests extend linearity testing. As a corollary of our main theorem, it follows that an ensemble of small graphs on overlapping sets of vertices can be glued together to one global graph assuming they agree with each other on average.

We prove the agreement theorem by (re)proving the agreement theorem for dimension $1$, and then generalizing it to higher dimensions (with the dimension $1$ case being the direct product test, and dimension $2$ being the graph case). A key technical step in our proof is the reverse union bound, which allows us to treat dependent events as if they are disjoint, and may be of independent interest. An added benefit of the reverse union bound is that it can be used to show that the ``majority decoded'' function also serves as a global function that explains the local consistency of the agreement theorem, a fact that was not known even in the direct product setting (dimension $1$) prior to our work.

Beyond the motivation to understand fundamental local-to-global structures, our main theorem allows us to lift structure theorems from $\mu_{1/2}$ to $\mu_{p}$.
As a simple demonstration of this paradigm, we show how the low degree testing result of of Alon~\etal~[\emph{IEEE Trans. Inform. Theory}, 2005] and Bhattacharyya~\etal~[\emph{Proc. 51st FOCS}, 2010], originally proved for $\mu_{1/2}$, can be extended to the biased hypercube $\mu_{p}$, even for very small sub-constant $p$.
\end{abstract}

\end{titlepage}

\pagenumbering{arabic}

\section{Introduction}

Agreement tests are a type of PCP tests and capture a fundamental local-to-global phenomenon.
In this paper, we study an agreement testing question that is a new extension of direct product testing to higher dimensions.

It is a basic fact of computation that any global computation can be broken down into a sequence of local steps. The PCP theorem  \cite{AroraS1998,AroraLMSS1998} says that moreover, this can be done in a robust fashion, so that as long as \emph{most} steps are correct, the entire computation checks out. At the heart of this is a \emph{local-to-global} argument that allows deducing a global property from local pieces that fit together only approximately.

A key example is the line vs.\ line \cite{GemmellLRSW1991,RubinfeldS1996} low degree test in the proof of the PCP theorem. In the PCP construction, a function on a large vector space is replaced by an ensemble of (supposed) restrictions to all possible affine lines. These restrictions are supplied by a prover and are not a~priori guaranteed to agree with any single global function. This is taken care of by the  ``low degree test'', which checks that restrictions on intersecting lines agree with each other, i.e.\ they give the same value to the point of intersection. The crux of the argument is the fact that the local agreement checks \emph{imply} agreement with a single global function. Thus, the low degree test captures a local-to-global phenomenon.

In what other scenarios does such a local-to-global theorem hold? This question was first asked by Goldreich and Safra \cite{GoldreichS2000}, who studied a combinatorial analog of the low degree test. Let us describe the basic framework of agreement testing  in which we will study this question.
In agreement testing, a global function is given by an ensemble of local functions.
There are two key aspects of agreement testing scenarios:
\begin{itemize}
\item Combinatorial structure: for a given ground set $V$ of size $n$, the combinatorial structure is a collection $H$ of subsets $S\subset V$ such that for each $S\in H$ we get a local function. For example, if $V$ is the points of a vector space then $H$ can be the collection of affine lines.
\item Allowed functions: for each subset $S\in H$, we can specify a space $\F_S$ of functions on $S$ that are allowed. The input to the agreement test is an ensemble of functions $\{f_S\}$ such that for every $S\in H$, $f_S\in \F_S$. For example, in the line vs.\ line low degree test we only allow local functions on each line that have low degree.
\end{itemize}

Given the ensemble $\set{f_S}$, the intention is that $f_S$ is the restriction to $S$ of a global function $F\colon V\to\Sigma$. Indeed, a local ensemble is called \emph{global} if there is a global function $F\colon V\to\Sigma$ such that
\[ \forall S\in H, \qquad f_S = F|_S.\]

An \emph{agreement check} for a pair of subsets $S_1,S_2$ is specified by a triple of sets $(T,S_1,S_2)$ such that $ T \subset S_1 \cap S_2$ and is denoted by $f_{S_1}\sim_T f_{S_2}$.  It checks whether the local functions $f_{S_1}$ and $f_{S_2}$ agree on the intersection $T$.
Formally,
\[ f_{S_1}\sim_T f_{S_2} \qquad \Longleftrightarrow \qquad
\forall x\in T, \quad f_{S_1}(x) = f_{S_2}(x).\]

A local ensemble which is global passes all agreement checks. The converse is also true: a local ensemble that passes \emph{all} agreement checks must be global.

An \emph{agreement test} is specified by giving a distribution $\mathcal{D}$ over triples of subsets $(T,S_1,S_2)$ such tha $T \subset S_1 \cap S_2$. We define the agreement of a local ensemble to be the probability of agreement:
\[ \agree_\mathcal{D}(\set{f_S}) := \Prob[(T,S_1,S_2)\sim D]{ f_{S_1}\sim_T f_{S_2} }.
\]
An agreement theorem shows that if $\set{f_S}_S$ is a local ensemble with $agree_\mathcal{D}(\set{f_S}) > 1-\eps$ then it is close to being global.

\paragraph{Example: direct product tests} Perhaps the simplest agreement test to describe is the direct product test, in which $H$ contains all possible $k$-element subsets of $V$. For each $S$, we let $\F_S$ be all possible functions on $S$, that is $\F_S = \{f\colon S\to\Sigma\}$. The input to the test is an ensemble of local functions $\{f_S\}$, and a natural testing distribution is to choose a random $t$-element set $T$ and two $k$-element subsets $S_1,S_2$ so that they intersect on $S_1 \cap S_2 = T$, a distribution we denote $\mu_n(k,t)$. Suppose $\agree(\set{f_S}) \ge 1-\eps$. Is there a global function $F\colon V\to\Sigma$ such that $F|_S = f_S$ for most subsets $S$, where most is measured with respect to the uniform distribution $\nu_n(k)$ over $k$-elements sets? This is the content of the direct product testing theorem of Dinur and Steurer~\cite{DinurS2014-dp}.
\begin{theorem}[agreement theorem, dimension 1]\label{thm:agree1-intro}
For all positive integers $n, k, t$ satisfying $1 \leq t < k \leq n$ and alphabet
$\Sigma$, the following holds.
Let $\{f_S\colon S \to \Sigma \mid S \in \binom{[n]}{k}\}$ be an ensemble
of local functions satisfying
\[
 \Pr_{(T,S_1,S_2) \sim \mu_n(k,t)}[f_{S_1}|_{T} \neq f_{S_2}|_{T}] \leq  \epsilon.
\]
Then there exists a global function $G\colon [n] \to \Sigma$ that satisfies
\[\Pr_{S\sim \nu_n(k)}[f_S \neq G|_S] = O_\paramq\left(\eps\right),\] where $\paramq = t/k$.
\end{theorem}

The qualitatively strong aspect of this theorem is that in the conclusion, the global function agrees \emph{perfectly} with $1-O(\eps)$ of the local functions. Achieving a weaker result where perfect agreement $f_S= F|_S$ is replaced by approximate one $f_S\approx F|_S$ would be significantly easier but also less useful. Quantitatively, this is manifested in that the fraction of local functions that end up disagreeing with the global function $F$ is at most $O(\eps)$ and is \emph{independent of $n$ and $k$}. It would be significantly easier to prove a weaker result where the closeness is $O(k\eps)$ (via a union bound on the event that $F(i) = f_S(i)$).
This theorem is proven by Dinur and Steurer~\cite{DinurS2014-dp} by imitating the proof of the parallel repetition theorem \cite{Raz1998}. This theorem is also used as a component in the recent work on agreement testing on high dimensional expanders~\cite{DinurK2017}.

\subsection*{Our Results}
In order to motivate our extension of \cref{thm:agree1-intro}, let us describe it in a slightly different form. The global function $F$ can be viewed as specifying the coefficients of a linear form $\sum_{i=1}^n F(i) x_i$ over variables $x_1,\ldots,x_n$. For each $S$, the local function $f_S$ specifies the partial linear form only over the variables in $S$. This $f_S$ is supposed to be equal to $F$ on the part of the domain where $x_i=0$ for all $i\not\in S$. Given an ensemble $\{f_S\}$ whose elements are promised to agree with each other on average, the agreement theorem allows us to conclude the existence of a global linear function that agrees with most of the local pieces.

 This description naturally leads to the question of extending this to higher degree polynomials. Now, the global function is a degree $d$ polynomial with coefficients in $\Sigma$, namely $F = \sum_{T}F(T) x_T$, where we sum over subsets $T\subset [n]$, $|T|\leq d$. The local functions $f_S$ will be polynomials of degree $\leq d$, supposedly obtained by zeroing out all variables outside $S$. Two local functions $f_{S_1},f_{S_2}$ are said to agree, denoted $f_{S_1}\sim f_{S_2}$, if every monomial that is induced by $S_1\cap S_2$ has the same coefficient in both polynomials. Our new agreement theorem says that in this setting as well, local agreement implies global agreement:

\begin{theorem}[agreement theorem for high dimensions]\label{thm:agreed-intro}
For all positive integers $n, k, t, d$ satisfying $d \leq t < k \leq n$ and alphabet
$\Sigma$ the following holds.
Let $\{f_S\colon \bind{S} \to \Sigma \mid S \in \binom{[n]}{k}\}$ be an ensemble
of local functions satisfying
\[
 \Pr_{(T,S_1,S_2) \sim \mu_n(k,t)}\left[f_{S_1}|_{T} \neq f_{S_2}|_{T}\right] \leq  \epsilon,
\]
then the majority decoded function $G\colon \bnd \to \Sigma$ satisfies
\[ \Pr_{S\sim \nu_n(k)}\left[f_S \neq G|_S\right] \leq O_{d,\paramq}(\eps)\,,\] where $\paramq = t/k$ and the majority-decoded function $G$ is the one given by ``popular vote'', namely for each $A \in \bnd$ set $G(A)$ to be the most frequently occurring value among $\set{f_S(A) \mid S\supset A}$ (breaking ties arbitrarily).
\end{theorem}

For $d=1$, this theorem is precisely \cref{thm:agree1-intro} (but for the ``furthermore'' clause). The additional ``furthermore'' clause strengthens our theorem by naming the popular vote function as a candidate global function that explains most of the local functions. This addendum strengthens also \cref{thm:agree1-intro}.

Let us spell out how this theorem fits into the framework described above. The ground set is $V = \binom{[n]}{\le d}$, and the collection of subsets $H$ is the collection of all induced hypergraphs on $k$ elements. In particular, if we focus on $\Sigma=\{0,1\}$, we can view the local function of a subset $S\subset [n]$, $|S|=k$, as specifying a hypergraph on the vertices of $S$ with hyperedges of size up to $d$. The theorem says that if these small hypergraphs agree with each other most of the time, then there is a global hypergraph that they nearly all agree with.

For the special case of $d=2$ and $\Sigma=\{0,1\}$, we get an interesting statement about combining small pieces of a graph into a global one.
\begin{corollary}[agreement test for graphs]\label{cor:graphs}
For all positive integers $n, k, t$ satisfying $2 \leq t < k \leq n$ the following holds.

Let $\{G_S\}$ be an ensemble
of graphs, where $S$ is a $k$ element subset of $[n]$ and $G_S$ is a graph on vertex set $S$. Suppose that
\[
 \Pr_{\substack{S_1,S_2 \in \bi{[n]}{k} \\ |S_1 \cap S_2| =
     t}}\left[G_{S_1}|_{S_1 \cap S_2} = G_{S_2}|_{S_1 \cap S_2}\right] \geq 1-  \epsilon.
\]
Then there exists a single global graph $G = ([n],E)$ satisfying $\Pr_{S \in
  \bi{[n]}{k}}[G_S = G|_S] =1- O_\paramq(\epsilon)$, where $\paramq = t/k$.
\end{corollary}
Here too we emphasize that the strength of the statement is in that the conclusion talks about \emph{exact} agreement between the global graph and the local graphs, i.e.\ $G_S = G|_S$ and not $G_S \approx G|_S$, for a fraction of $1-O(\eps)$ of the sets $S$. It is also important that there is no dependence in the $O(\cdot)$ on either $n$ or $k$.
A similar agreement testing statement can be made for hypergraphs of any uniformity $\le d$.

A technical component in our proof which we wish to highlight is the \emph{reverse union bound}, which may be of independent interest.

\begin{restatable}[reverse union bound]{lemma}{rub}\label{lem:rub}
For each $c > 0$ and integer $d \geq 1$ there exists a constant $C'(d,c) = 2^{O_c(d^3)}$ such that the following holds for all $d \leq k \leq n$. Let $Y \subseteq \binom{[n]}{k}$ and $B \subseteq \binom{[n]}{d}$ satisfy
\[
 \Pr_{S \sim \nu_n(k)}[S \in Y \mid S \supseteq I] \geq c \text{ for all } I \in B\, .
\]
Then
\[
 \Pr_{S \sim \nu_n(k)}[S \text{ hits } B] \leq C'(c,d)\cdot \Pr_{S \sim \nu_n(k)}[S \in Y] \, ,
\]
where $S$ \emph{hits} $B$ if $S \supseteq I$ for some $I \in B$.
\end{restatable}

We illustrate an application of this lemma later on in the introduction.

\subsection*{Application to Structure Theorems}

The agreement theorems we prove allow for a new paradigm for studying the structure of Boolean functions on the \emph{biased} Boolean hypercube, i.e.\ when the measure is $\mu_p$ and $p$ is potentially very small, e.g. $p = o_n(1)$, $p \to 0$ as $n \to \infty$.

This paradigm is based on the following simple fact: the $p$-biased hypercube is expressible as a convex combination of many small-dimensional copies of the uniform hypercube. To uncover structure for $\mu_p$, we invoke known structure theorems for $\mu_{1/2}$, obtaining a structured approximation for each copy separately. We then sew these approximations together using the agreement theorem.
This strategy allows us to lift structure theorems from $\mu_{1/2}$ to $\mu_{p}$.

\paragraph{$GF(2)$-low degree testing in the biased hypercube}

As an illustration of this paradigm, we lift the low degree test~\cite{AlonKKLR2005,BhattacharyyaKSSZ2010} to the biased setting using a straightforward application of the agreement theorem. This is similar to the way in which the analysis of the ``uniform BLR'' test was lifted from the middle slice to an arbitrary slice by David~\etal~\cite{DavidDGKS2017}.

Alon~\etal~\cite{AlonKKLR2005} studied a $2^{d+1}$-query test $T_d$ for being of low degree. Bhattacharyya~\etal~\cite{BhattacharyyaKSSZ2010} gave an optimal analysis of this test, showing that $\delta_d(f) = O_d(\rej_d(f))$, where $\delta_d(f)$ refers to the distance of $f$ to the closest degree $d$ function under the $\mu_{1/2}$ measure (in other words, $\delta_d(f) = \min\limits_{\bdeg(g) \leq d}\Pr_{\mu_{1/2}}[f\neq g]$), and $\rej_d(f)$ is the rejection probability of the test $T_d$ on input function $f$. We would like to extend the test $T_d$ to the $p$-biased setting, wherein we measure closeness of $f$ to low degree functions with respect to the $\mu_p$ measure instead of $\mu_{1/2}$ measure. More precisely, $\delta^{(p)}_d(f) := \min\limits_{\bdeg(g) \leq d} \Pr_{\mu_p}[f \neq g].$ To this end, we study the following test $T_{p,d}$. \\

\begin{restatable}[H]{algorithm}{tpd}  \setstretch{1.35}
  \SetKwInOut{Input}{Input}
  \Input{$f\colon \bits^n \to \bits$ (provided as an oracle)}
  \BlankLine
  Pick $S \subseteq [n]$ according to the distribution $\mu_{2p}$.

  If $|S| \leq d$, then accept.

  Let $f|_S\colon \bits^S \to \bits$ denote the restriction of $f$ to $\bits^S$ by zeroing out all the coordinates outside $S$.

  Pick $x, a_1, \dots,a_{d+1} \in \bits^S$ independently from the distribution $\mu_{1/2}^{\otimes S}$, subject to the constraint that $a_1,\ldots,a_{d+1}$ are linearly independent.

  Accept iff \[\sum_{I \subseteq [d+1]} f|_S\left(x + \sum_{i \in I}a_i\right) = 0 \pmod{2}\enspace. \]

    \caption{\textsc{$p$-biased Low Degree Test $T_{p,d}$}}\label{alg:Tpd}
  \end{restatable}
\;

We use the agreement theorem to show that this natural extension is a valid low degree test for the $p$-biased setting.

\begin{restatable}[$p$-biased version of the BKSSZ Theorem]{theorem}{pBKSSZ}\label{thm:bkssz-p}
For every $d$ and $p \in (0,\nicefrac12)$, the $2^{d+1}$-query test $T_{p,d}$ satisfies the following properties:
\begin{itemize}
\item Completeness: if $\bdeg(f) \leq d$ then $\rej_{T_{d,p}}(f) = 0$.
\item Soundness: $\delta_d^{(p)}(f) = O_d(\rej_{T_{d,p}}(f))$, where the hidden constant is independent of $p$ (but depends on $d$).
\end{itemize}	
\end{restatable}

\paragraph{Kindler--Safra Junta theorem for the biased hypercube}

A second (and more involved) illustration of this paradigm is lifting the Kindler--Safra Junta theorem to the $p$-biased hypercube, even when $p$ is very small, proved by the authors~\cite{DinurFH-ks} in a parallel work to the current paper. 
Previous proofs of such statements deteriorated as $p$ becomes smaller due to the use of hypercontactivity.

The application to the Kindler-Safra Junta theorem for the $p$-baised hypercube served as the primary motivation for the agreement theorem proved in this work. We remark that while the original version of \cite{DinurFH-ks}  used \cref{thm:agreed-intro}, a subsequent revision observed that a \emph{simpler} agreement theorem that worked only for juntas sufficed. The revised version of \cite{DinurFH-ks} contains this simpler ``junta'' agreement theorem with a self-contained proof.

\subsection*{Context and Motivation}
Agreement tests were first studied implicitly in the context of PCP theorems. In fact, every PCP construction that has a composition step invariably relies on an agreement theorem. This is because in a typical PCP construction, the proof is broken into small pieces that are further encoded e.g.\ by composition or by a gadget. The soundness analysis decodes each gadget separately, thereby obtaining a collection of local views. Then, essentially through an agreement theorem, these are stitched together into one global NP witness. Similar to locally testable codes, agreement tests are a combinatorial question that is related to PCPs. Interestingly, this relation has recently been made formal by Dinur~\etal~\cite{DinurKKMS2018-2to1}, where it is proved that a certain agreement test (whose correctness was proved by Khot et al.~\cite{KhotMS2018}) \emph{formally implies} a certain rather strong unique games PCP theorem. Such a formal connection is not known to exist between LTCs and PCPs. For example, even if someone manages to construct the ``ultimate'' locally testable codes with linear length and distance, and testable with a constant number of queries, this is not known to have any implications for constructing linear size PCPs (although one may hope that such codes will be useful toward that goal).

Beyond their role in PCPs, we believe that agreement tests capture a fundamental local-to-global phenomenon, and merit study on their own. Exploring new structures that support agreement theorems seems to be an important direction. Several steps in this direction have been taken by Dinur and Kaufman~\cite{DinurK2017} and by Dikstein and Dinur~\cite{DiksteinD2019}.

\paragraph{Relation to Property Testing}
Agreement testing is similar to property testing in that we study the relation between a global object and its local views.
In property testing we have access to a single global object, and we restrict ourselves to look only at random local views of it. In agreement tests, we don't get access to a global object, but rather to an \emph{ensemble of local functions} that are not apriori guaranteed to come from a single global object. Another difference is that unlike in property testing, in an agreement test the local views are pre-specified and are a part of the problem description, rather than being part of the algorithmic solution.

Still, there is an interesting interplay between \cref{cor:graphs}, which talks about combining an ensemble of local graphs into one global graph, and graph property testing. Suppose we focus on some testable graph property, and suppose further that the test proceeds by choosing a random set of vertices and reading all of the edges in the induced subgraph, and checking that the property is satisfied there (many graph properties are testable this way, for example bipartiteness \cite{GoldreichGR1998}). Suppose we only allow ensembles $\{G_S\}$ where for each subset $S$, the local graph $G_S$ satisfies the property (e.g.\ it is bipartite). This fits into our formalism by specifying the space of allowed functions $\F_S$ to consist only of accepting local views. This is analogous to requiring, in the low degree test, that the local function on each line has low degree as a univariate polynomial. By \cref{cor:graphs}, we know that if these local graphs agree with each other with probability $1-\eps$, there is a global graph $G$ that agrees with $1-O(\eps)$ of them. In particular, this graph \emph{passes the property test}, so must itself be close to having the property! At this point it is absolutely crucial that the agreement theorem provides the stronger guarantee that $G|_S=G_S$ (and not $G|_S\approx G_S$) for $1-O(\eps)$ of the $S$'s. We can thus conclude that not only is there a global graph $G$, but actually that this global $G$ is close to having the property.

This should be compared to the low degree agreement test, where we only allow local functions with low degree, and the conclusion is that there is a global function that itself has low degree.

\subsection*{Technical Contribution}

Instead of proving \cref{thm:agreed-intro} directly, we first prove a version which employs a different test distribution $\nu_n(k,t)$. For comparison, here are the two distributions:
\begin{itemize}
\item $(T,S_1,S_2) \sim \mu_n(k,t)$ if $T$ is a random set of size $t$, and $S_1,S_2$ are two random extensions of $T$ of size $k$ such that $S_1 \cap S_2 = T$.
\item $(T,S_1,S_2) \sim \nu_n(k,t)$ if $T$ is a random set of size $T$, and $S_1,S_2$ are two random extensions of $T$ of size $k$.
\end{itemize}
In both cases, the test is $f_{S_1}|_T = f_{S_2}|_T$. For $\mu_n(k,t)$, this is the same as $f_{S_1}|_{S_1 \cap S_2} = f_{S_2}|_{S_1 \cap S_2}$, but for $\nu_n(k,t)$ it could be that $S_1 \cap S_2 \supsetneq T$.

We deduce \cref{thm:agreed-intro} from its $\nu_n(k,t)$-version via a coupling argument.

\smallskip

To explain the proof, we start with the easier case $d=1$, which is just the $\nu_n(k,t)$-analog of \cref{thm:agree1-intro}.
Given an ensemble $\{f_S\}$, it is easy to define the global function $G$, by popular vote (``majority decoding''). The main difficulty is to prove that for a typical set $S$, $f_S$ agrees with $G|_S$ on all elements $i\in S$ (and later on all $d$-sets).

Our proof doesn't proceed by defining $G$ as majority vote right away. Instead, like in many previous proofs~\cite{DinurG2008,ImpagliazzoKW2012,DinurS2014-dp}, we condition on a certain event (focusing, say, on all subsets that contain a certain set $T$, and such that $f_S|_T = \alpha$ for a certain value of $\alpha$), and define a ``restricted global'' function, for each $T$, by taking majority just among the sets in the conditioned event. This boosts the probability of agreement inside this event. After this boost, we can afford to take a union bound and safely get agreement with the restricted global function $G_T$. The proof then needs to perform another agreement step which stitches the restricted global functions $\set{G_T}_T$ into a completely global function. The resulting global function does not necessarily equal the majority vote function $G$, and a separate argument is then carried out to show that the conclusion is correct also for $G$.

In higher dimensions $d>1$, these two steps of agreement (first to restricted global and then to global) become a longer sequence of steps, where at each step we are looking at restricted functions that are defined over larger and larger parts of the domain.

The main technical difficulty is that a single event $f_S = F|_S$ consists of $\binom{k}{d}$ little events, namely $f_S(A) = F(A)$ for all $A\in \binom S d$, that each have some probability of failure. We thus need to boost the failure probability from $\eps$ to $\eps/k^d$
so that we can afford to take a union bound on the $\binom k d$ different sub-events. How do we get this large boost? Our strategy is to proceed in steps, where at each stage, we condition on the global function from the previous stage, boosting the probability of success further.

\paragraph{Majority decoding} The most natural choice for the global function $F$ in the conclusion of \cref{thm:agreed-intro} is the majority decoding, where $F(A)$ is the most common value of $f_S(A)$ over all $S$ containing $A$. This is the content of the ``furthermore'' clause in the statement of the theorem. Neither the proof strategy of Dinur and Steurer~\cite{DinurS2014-dp} nor our generalization promises that the produced global function $F$ is the majority decoding.

 Our strategy produces a global function which agrees with most local functions, but we cannot guarantee immediately that this global function corresponds to majority decoding. What we are able to show is that \emph{if} there is a global function agreeing with most of the local functions \emph{then} the function obtained via majority decoding also agrees with most of the local functions. We outline the argument below.

Suppose that $\{ f_S \}$ is an ensemble of local functions that mostly agree with each other, and suppose that they also mostly agree with some global function $F$. Let $G$ be the function obtained by majority decoding: $G(A)$ is the most common value of $f_S(A)$ over all $S$ containing $A$. Our goal is to show that $G$ also mostly agrees with the local functions, and we do this by showing that $F$ and $G$ mostly agree.

Suppose that $F(A) \neq G(A)$. We consider two cases. If the distribution of $f_S(A)$ is very skewed toward $G(A)$, then $f_S(A) \neq F(A)$ will happen very often. If the distribution of $f_S(A)$ is very spread out, then $f_{S_1}(A) \neq f_{S_2}(A)$ will happen very often. Since both events $f_S(A) \neq F(A)$ and $f_{S_1}(A) \neq f_{S_2}(A)$ are known to be rare, we would like to conclude that $F(A) \neq G(A)$ happens for very few $A$'s.

Here we face a problem: the bad events (either $f_S(A) \neq F(A)$ or $f_{S_1}(A) \neq f_{S_2}(A)$) corresponding to different $A$'s are not necessarily disjoint. A~priori, there might be many different $A$'s such that $F(A) \neq G(A)$, but the bad events implied by them could all coincide.

The missing piece is the reverse union bound, \cref{lem:rub}. We choose $Y = \{ S : f_S \neq F|_S \}$ and $B = \{ A : F(A) \neq G(A) \}$. If $A \in B$ then the probability that $f_S(A) \neq F(A)$ (over $S \supseteq A$) is at least $1/2$, since otherwise $F(A)$ would have been the majority decoded value. Therefore the premise of \cref{lem:rub} is satisfied (with $c = 1/2$), and we conclude that the probability of $F|_S \neq G|_S$, which is the same as the probability that $S$ hits $B$, is only $O(\Pr[f_S \neq F|_S]) = O(\epsilon)$.

\paragraph{Reverse union bound} The reverse union bound follows from a \emph{hypergraph pruning lemma}, which allows approximating a given hypergraph $H$ by a subhypergraph $H'\subset H$ which is sparse and has certain ``disjointness'' properties.

\begin{restatable}[hypergraph pruning lemma]{lemma}{hpl}\label{lem:pruning-uniform}
For each $\epsilon > 0$ and integer $d \geq 1$ there exists a constant $C''(d,\epsilon) = 2^{O_\epsilon(d^3)}$ such that for every $d$-uniform hypergraph $H$ there exists a $d$-uniform hypergraph $H' \subseteq H$ such that:
\begin{itemize}
\item $\Pr_{S\sim\nu_n(k)}[S \text{ hits } H'] \geq \Pr_{S\sim\nu_n(k)}[S \text{ hits } H] / C''(d,\epsilon)	$.
\item For every $I \in H'$,
\[
 \Pr_{S\sim\nu_n(k)}[\not\exists I' \in H' \setminus \{I\}, S \supseteq I' \mid S \supseteq  I ] \geq 1 - \epsilon.
\]
\end{itemize}
\end{restatable}

The hypergraph pruning lemma states that we can prune a given $d$-uniform hypergraph $H$ to a subhypergraph $H'$ such that the events ``$S \supseteq  I$'' (where $S \sim \nu_n(k)$ and $I \in H'$) are nearly disjoint, while maintaining the property that the probability that $S$ contains some hyperedge in $H'$ does not drop too much compared to the corresponding probability with respect to $H$.

The pruned hypergraph $H'$ satisfies the following crucial property, for each $I \in H'$:
\[
 \Pr_{S \sim \nu_n(k)}[S \text{ hits } H' \text{ only at } I] = \Theta\left(\Pr_{S \sim \nu_n(k)}[S \supseteq I]\right).
\]
Crucially, the events ``$S \text{ hits } H' \text{ only at } I$'' are disjoint. This allows us to quickly prove \cref{lem:rub}: applying the hypergraph pruning lemma, we prune $B$ to a smaller hypergraph $B'$. The disjointness of the aforementioned events implies that
\begin{multline*}
 \Pr[S \in Y] \geq
 \sum_{I \in B'} \Pr[S \text{ hits } B' \text{ only at } I \text{ and } S \in Y] \\ \stackrel{(\ast)}\geq \Omega(1) \cdot  \sum_{I \in B'} \Pr_{S \sim \nu_n(k)}[S \supseteq I] = \Omega(\Pr[S \text{ hits } B']) = \Omega(\Pr[S \text{ hits } B]),
\end{multline*}
where $(\ast)$ follows from the crucial property:
\[
 \Pr[S \text{ hits } B' \text{ only at } I \text{ and } S \in Y] \geq
 \Pr[S \text{ hits } B' \text{ only at } I] - \Pr[S \supseteq I \text{ and } S \in Y] = \Omega(\Pr[S \supseteq I]).
\]

The proof of \cref{lem:pruning-uniform} is inductive and elementary. It is proved by introducing the notion of \emph{branching factor}; a hypergraph $H$ over a vertex set $V$ is said to have \emph{branching factor} $\rho$ if for all subsets $A\subset V$ and integers $r \geq 0$, there are at most $\rho^r$ hyperedges in $H$ of cardinality $|A|+r$ containing $A$ (see \cref{sec:rub} and \cref{def:bf}  therein for further details). 

\subsection*{Organization}
The rest of this paper is organized as follows. We start in \cref{sec:agree1} by proving the version of \cref{thm:agree1-intro} which uses the test distribution $\nu_n(k,t)$. We then generalize this to higher dimensions in \cref{sec:agreed}.
We deduce \cref{thm:agreed-intro} in \cref{sec:agreement-other}, in which we also deduce a $p$-biased agreement theorem and an auxiliary agreement theorem required in the subsequent section.
The application to $GF(2)$-low degree testing appears in \cref{sec:reed-muller}. We conclude the paper by proving the reverse union bound in \cref{sec:rub}.

\section{Notation and Preliminaries}

We begin with some basic notation. Let $[n] = \{1,\dots, n\}$. For $k \leq n$, we will refer to a set $S \subseteq [n]$ of size $k$ as a $k$-set of $[n]$.

We will be working with several distributions on subsets of $[n]$ (and pairs, triples of subsets on $[n]$), which are defined in detail below.
\begin{description}

\item[$\nu_n(k)$:] The \emph{slice} $\binom{[n]}{k}$ consists of all $k$-subsets of $[n]$. We denote the uniform distribution over $\binom{[n]}{k}$ by $\nu_n(k)$.

\item[$\nu_n(k,t)$:] For $1 \leq t \leq k \leq n$, let $\nu_n(k,t)$ denote the distribution over triples of sets $(T,S_1,S_2)$ obtained as follows:  pick a uniformly random $t$-set $T$ of $[n]$ and then pick independently two random $k$-sets $S_1,S_2$ of $[n]$ such that $T \subset S_1, S_2$.

\item[$\mu_n(k,t)$:] The distribution $\mu_n(k,t)$, on the other hand, is the distribution over pair of sets $(S_1,S_2) $ obtained as follows: pick a uniformly random a $t$-set $T$ and then two $k$-sets $S_1$ and $S_2$  uniformly at random conditioned on $S_1 \cap S_2 = T$. We will refer to the latter distribution $\mu_n(k,t)$ as the exact-intersection distribution.

  We now define the $p$-biased version of the above distributions.

\item[$\mu_p(A)$:] For a set $A$, the distribution $\mu_p(A)$ denotes a random subset of $A$ obtained by putting it each element with probability $p$ independently. Stated differently, $\Pr[\mu_p(A) = B] = p^{|B|} (1-p)^{|A \setminus B|}$.

\item[$\mu_{p,\paramq}(A)$:] The distribution $\mu_{p,\paramq}(A)$, defined whenever $p(2-\paramq) \leq 1$, is a distribution on pairs of subsets $B_1,B_2 \subseteq A$ which are sampled as follows: for each element $x \in A$, we put $x$ in both $B_1$ and $B_2$ with probability $p\paramq$; we put $x$ only in $B_1$ or only in $B_2$ with probability $p(1-\paramq)$ each; and we put $x$ in neither with probability $1-p(2-\paramq)$.

\end{description}

\subsection{Plurality vs.\ collision}

Let $X$ be a random variable ranging over a finite set, and let $x_0$ be the most probably value. Let $Y$ be an i.i.d.\ copy of $X$. Then
\[
 \Pr[X \neq x_0] \leq \Pr[X \neq Y].
\]

To see this, let $p_x$ be the probability of $x$. Then
\[
 \Pr[X \neq Y] = \sum_x p_x (1-p_x) \geq \sum_x p_x (1-p_{x_0}) = 1 - p_{x_0} = \Pr[X \neq x_0].
\]

\subsection{Reverse union bound}

We recall the statement of the reverse union bound from the introduction.
\rub*

We prove \Cref{lem:rub} in \Cref{sec:rub}. For now, let us explain the name \emph{reverse union bound}. One could try to lower bound $\Pr[S \in Y]$ in terms of $\Pr[S \text{ hits } B]$ using the following simple argument: if $S$ hits $B$, then it contains some $I \in B$, and so the probability that $S \in Y$ should be at least $c$. In symbols, this would be
\[
 \Pr[S \in Y] \stackrel{(\ast)}\geq \sum_{I \in B} \Pr[S \supseteq I] \Pr[S \in Y \mid S \supseteq I] \geq c \sum_{I \in B} \Pr[S \supseteq I] \geq c\Pr[S \text{ hits } B].
\]
This argument is, of course, fallacious: $(\ast)$ fails since the events in question are not disjoint. The reverse union bound shows that if you are willing to pay in the constant $c$, then you can make sense of this fallacious argument.

\section{One-dimensional agreement theorem}\label{sec:agree1}

In this section, we prove the following direct product agreement
testing theorem for one dimension. This theorem
is a special case of the more general theorem (\cref{thm:agreed})
proved in the next section, and a slight strengthening of the result of Dinur
and Steurer~\cite{DinurS2014-dp} (\cref{thm:agree1-intro}). We give the proof for the
one-dimensional case as it serves as a warmup to the general
case.

\begin{theorem}[agreement theorem, dimension 1 (strengthened version of \cref{thm:agree1-intro})]\label{thm:agree1}
For all positive integers $n, k, t$ satisfying $1 \leq t < k \leq n$ and alphabet
$\Sigma$, the following holds.
Let $\{f_S\colon S \to \Sigma \mid S \in \binom{[n]}{k}\}$ be an ensemble
of local functions satisfying
\[
 \Pr_{(T,S_1,S_2) \sim \nu_n(k,t)}[f_{S_1}|_{T} \neq f_{S_2}|_{T}] \leq  \epsilon.
\]
Then the majority-decoded function $G\colon [n] \to \Sigma$ satisfies
\[\Pr_{S\sim \nu_n(k)}[f_S \neq G|_S] = O\left(\frac{\eps}{\min\left(\frac{t}{k},1-\frac{t}{k}\right)}\right),\] where
the majority-decoded function $G$ is the one given by ``popular vote'', namely for each $i \in [n]$, $G(i)$ is the most frequently occurring value among $\set{f_S(i) \mid S\ni i}$ (breaking ties arbitrarily).
\end{theorem}

\subsection{Two examples}

Consider the following two examples.
\paragraph{Example 1: ``Small non-expanding set of $k$-tuples''} Let $B$ be an arbitrary set of size $(\epsilon/2k)n$. We define $f_S$ as follows: $f_S(i) = 1_{S \cap B \neq \emptyset}$. Majority decoding gives the function $g(i) = 0$.  In this example, the probability that a random $k$-set $S_1$ intersects $B$ is $\Theta(\epsilon)$, and when that happens, the intersection with $B$ has constant size. Therefore, given that $S_1$ intersects $B$, the probability that another random $k$-set $S_2$ also intersects $B$ is roughly $t/k$. Overall, the probability that $f_{S_1},f_{S_2}$ disagree is $\Theta((1-t/k)\epsilon)$. In contrast, $\Pr[f_S \neq g|_S] = \Theta(\epsilon)$.
\paragraph{Example 2: ``Bernoulli noise''} Consider the following function: $f_S(i) \sim \operatorname{Ber}(\epsilon/2t)$, where majority decoding again gives the function $g(i) = 0$. In this example, the probability that $f_{S_1}$ and $f_{S_2}$ disagree is $\Theta(\epsilon)$, while $\Pr[f_S \neq g|_S] \approx 1 - \exp(-(k/2t)\epsilon)$, which is close to $\Theta((k/t)\epsilon)$ for $\epsilon$ much smaller than $k/t$.\\

\noindent These two examples show that the statement of \cref{thm:agree1} is tight.

We begin by considering a natural approach\footnote{This argument is adapted from the introduction of Yotam Dikstein's Master's thesis~\cite{Dikstein2019}} to proving the above theorem using the majority (or plurality) decoding, and show how this approach only obtains an error bound of $O(k \epsilon)$ instead of $O(\epsilon)$ as promised by the above result.

For the purpose of this section, we assume that $0.1 \leq t/k \leq 0.5$. Given a local ensemble of functions $\{f_S \colon S \to \Sigma \mid S \in \binom{[n]}k\}$, define the majority decoded function $g \colon [n] \to \Sigma$ as follows: $g(i)$ is the most popular value of $f_S(i)$ among sets $S$ containing $i$. Plurality vs.\ collision implies that for each $i \in [n]$, \[
 \Pr_{S \ni i}[f_S(i) \neq g(i)] \leq \Pr_{S_1,S_2 \ni i}[f_{S_1}(i) \neq f_{S_2}(i)].
\]
The distribution of the pair of sets $S_1, S_2$ in the above expression is not the same as that used in the test. To relate this to the success probability of the test, we construct the following coupled distribution $(T_1,T_2,S_1,S_2, S)$ for each $i \in [n]$: choose two $t$-sets $T_1,T_2$ containing $i$, choose two random $k$-sets $S_1$ and $S_2$ such that $T_1 \subset S_1$ and $T_2 \subset S_2$, and finally choose a random $k$-set $S$ such that $T_1 \cup T_2 \subset S$. (The last step is possible if $1+2(t-1) = 2t-1 \leq k$, which is true since $0.1 \leq t/ k \leq 0.5$). This coupling has the following property: the pair of sets $(S_1,S_2)$ are two random $k$-sets containing $i$, while for each $j \in [2]$ the pair of sets $(S_j,S)$ are two random $k$-sets containing $T_j$, which is itself a random $t$-set containing $i$.

If $f_{S_1}(i) \neq f_{S_2}(i)$ then $f_{S_j}(i) \neq f_S(i)$ for some $j \in [2]$, and so
\[
 \Pr_{S \ni i}\left[f_S(i) \neq g(i)\right]\leq \Pr_{S_1,S_2 \ni i}[f_{S_1}(i) \neq f_{S_2}(i)] \leq 2\Pr_{\substack{T \ni i \\ S_1,S_2 \supset T}}\left[f_{S_1}(i) \neq f_{S_2}(i)\right].
\]
Undoing the conditioning, and using the fact that for any fixed $i$, $\Pr\left[S\ni i \right] = \Theta(\Pr\left[ T\ni i\right])$ (since  $0.1 \leq t/k \leq 0.5$), this shows that
\[
 \Pr_S\left[S\ni i  \text{ and } f_S(i) \neq g(i)\right] \leq O\left(\Pr_{T; S_1,S_2 \supset T}\left[T\ni i \text{ and } f_{S_1}(i) \neq f_{S_2}(i)\right]\right).
\]
Summing over all $i$, this gives
\begin{align*}
 \Pr_S\left[f_S \neq g|_S\right] &\stackrel{(1)}\leq \sum_{i=1}^n \Pr_S[S\ni i \text{ and } f_S(i) \neq g(i)] \\ &=
 O\left(\sum_{i=1}^n \Pr_{T; S_1,S_2 \supset T}\left[ T \ni i\text{ and } f_{S_1}(i) \neq f_{S_2}(i)\right] \right) &&=                                                                               O\left(\EE_T \left[\sum_{T\ni i} \Pr_{S_1,S_2 \supset T}\left[f_{S_1}(i) \neq f_{S_2}(i)\right] \right]\right)\\
  & =
 O\left(\EE_{T; S_1,S_2 \supset T}\left[\#\{ T\ni i \mid  f_{S_1}(i) \neq f_{S_2}(i)\}\right]\right) &&\stackrel{(2)}\leq O(t\epsilon).
\end{align*}

\medskip

Observe that this simple argument gives a bound which is off by a multiplicative factor of $t$. Inequality $(1)$ is loose whenever $f_S\neq g|_S$ on many coordinates, and inequality $(2)$ is loose whenever $f_{S_1},f_{S_2}$ disagree on a few coordinates.

Can it really be that conditioned on $f_{S_1} \neq f_{S_2}$, the number of bad $i$, i.e, the number of $i\in T$ such that $f_{S_1}(i) \neq f_{s_2}(i)$, be as large as $t$? Yes! In Example 1, when $f_{S_1} \neq f_{S_2}$, it is always the case that $f_{S_1},f_{S_2}$ disagree on all of their intersection. Hence (2) is tight, however (1)~is loose: the contribution of $i \in B$ to the sum is $\epsilon/2$, while each of the $(1-\epsilon/(2k))n$ other indices contributes roughly $(\epsilon/2)(k/n)$, for a total of roughly $(\epsilon/2)k$. On the other hand, in example 2, ~(1) is tight but~(2) is loose, since the expected number of bad $i$ is constant. Our main result shows that (1)~and~(2) cannot both be tight at the same time.

\subsection{Proof of {\cref{thm:agree1}}}

Let $t+1 \leq k \leq n$ as in the hypothesis.
We will assume the following convention throughout this section.
\begin{itemize}
	\item $S$ will always denote a set of size $k$.
	\item $T$ will always denote a set of size $t$.
	\item $U$ will always denote a set of size $t-1$.
	\item $n' = n-(t-1)$ and $k' = k-(t-1)$.
        \end{itemize}
Define
\begin{align*}
  \zeta &:= t/(k-t+1) = t/k' \text{ and } \xi := 1-t/k .
\end{align*} Given a global function $g\fdom{[n]}$ and a local function $f_S\fdom{S}$, we say that $f_S$ and $g$ {\em disagree} if $f_S \neq g|_S$.

We now turn to the actual proof. Let $U$ be a $(t-1)$-set. In the first step (which we refer to as the relative decoding step), we define a ``global'' function $g_U\colon [n] \to \Sigma$ relative to the set $U$ and show that for a typical $(t-1)$-set $U$ and $k$-set $S$ such that $U \subset S$, we have $\Pr_{U, S\supset U}[f_S \neq g_U|_S] = O(\eps)$. In the next step (which we refer to as the sewing step), we sew together these different relative functions $g_U$ and show that there exists a $U$ such that for most $k$-sets $S$ (not necessarily $k$-sets that contain $U$), we have that $f_S = g_U|_S$. In the final step, we conclude that if such a ``global'' $g_U$ exists, then the majority decoded function also works.

\subsubsection{Part 1: relative decoding}

Let $U$ be a set of size $t-1$. We define a function $g=g_U\colon [n] \to \Sigma$ as follows (we drop the subscript $U$ as $U$ is fixed  for most of the argument below).
\begin{enumerate}
\item $g|_U$ is the most popular value of $f_S|_U$ among sets $S$ containing $U$.
\item For $i \notin U$, $g(i)$ is the most popular value of $f_S(i)$ among sets $S$ containing $U$ and satisfying $f_S|_U = g|_U$.
\end{enumerate}
Our goal in this part is to bound $\EE_U \Pr_{S \supset U}\left[f_S,g \text{ disagree}\right]$. More precisely,

\begin{lemma}\label{lem:agree1-local}
For all positive integers $n, k, t$ satisfying $1 \leq t < k \leq n$ and alphabet
$\Sigma$ the following holds.
Let $\{f_S\colon S \to \Sigma \mid S \in \binom{[n]}{k}\}$ be an ensemble
of local functions satisfying
\[
 \Pr_{(T,S_1,S_2) \sim \nu_n(k,t)}\left[f_{S_1}|_{T} \neq f_{S_2}|_{T}\right] \leq  \epsilon.
\]
Then there exists an ensemble $\{g_U\colon [n] \to \Sigma \mid U \in
\binom{[n]}{t-1}\}$ of global functions such that when a random $U \in
\binom{[n]}{t-1}$ and $S \in \binom{[n]}{k}$ are
chosen satisfying $S \supset U$, we have
\[\Pr_{U, S \supset U}\left[g_U|_S \neq f_S\right] =
O(\eps/\zeta),\] where $\zeta = t/(k-t+1)$.
\end{lemma}

\begin{proof}
We begin by defining the notion of a {\em good} element $i$ and a {\em good} $k$-set $S$.  We say that $i \notin U$ is \emph{good} if
\[
 \Pr_{S \supset U+i}\left[f_S|_U = g|_U\right] \geq \frac{1}{2}.
\]
We say that a $k$-set $S \supset U$ is \emph{good} if all $i \in S \setminus U$ are good.

We bound the quantity of interest $\Pr_{S \supset U}\left[f_S \neq g|_S\right]$ for a fixed $U$  as follows:
\begin{align*}
  \Pr_{S \supset U}\left[f_S \neq g|_S\right] & \leq \Pr_{S \supset U}\left[ S \text{ is bad}\right] + \Pr_{S \supset U}\left[ S \text{ is good  and } f_S \neq g|_S\right]\\
  & \leq \Pr_{S \supset U}\left[ S \text{ is bad}\right] + \Pr_{S \supset U}\left[f_S|_U \neq g|_U\right]  + \Pr_S\left[ S \text{ is good  and } f_S |_U= g|_U \text{ and } \exists i \in S \setminus U, f_S(i) \neq g(i)\right]\\
  & \leq \Pr_{S \supset U}\left[ S \text{ is bad}\right]+ \Pr_{S \supset U}\left[f_S|_U \neq g|_U\right] + \sum_{\substack{i \notin U \\ i \text{ good}}}\Pr_{S \supset U}\left[f_S|_U =  g|_U \text{  and }  i \in S \text{ and } f_S(i) \neq g(i)\right]\\
 & =  \underbrace{\Pr_{S \supset U}\left[\text{$S$ is bad}\right]}_{\delta} +\underbrace{\Pr_{S \supset U}\left[f_S|_U \neq g|_U\right]}_{\gamma} +
 \frac{k'}{n'} \sum_{\substack{i \notin U \\ i \text{ good}}}\underbrace{\Pr_{S \supset U+i}\left[f_S|_U = g|_U \text{ and } f_S(i) \neq g(i)\right]}_{\gamma_i}.
\end{align*}
The penultimate step follows from a union bound and the fact that for good $S$, every $i \in S \setminus U$ is good. The last step follows from the fact that $\Pr_{S \supset U}\left[ i \in S \right] = k'/n'$. Finally, we can extend the last summation to all $i$ (not necessarily good $i$) by defining $\gamma_i = 0$ for bad $i$. We bound each of the error terms $\gamma, \delta$ and $\gamma_i$ in the following three claims.

\begin{claim}[bounding $\gamma$] If $1 \leq t < k \leq n$, then $\EE_U\left[\gamma\right] \leq 2\epsilon.$\label{dlm:dim1gamma}\end{claim}
\begin{proof}By plurality vs.\ collision, $\gamma \leq \Pr_{S_1,S_2 \supset U}\left[f_{S_1}|_U \neq f_{S_2}|_U\right]$.

  Now consider the following coupling. We choose random independent $i_1,i_2 \notin U$, a random $k$-set $S_1 \supset U+i_1$, a random $k$-set $S_2 \supset U+i_2$, and a random $k$-set $S \supset U \cup \{i_1,i_2\}$. Here, we are using the fact that $k \geq (t-1)+2 = t+1$. Also, note $i_1$ might be equal to $i_2$ in the above coupling. Since $i_1,i_2$ are independent, the marginal distribution of $(S_1,S_2)$ is the same as two independent sets containing $U$, and the marginal distribution of $(S,S_j)$ is the same as two independent sets containing $U$ and one additional element.
Thus if $f_{S_1}|_U \neq f_{S_2}|_U$, it must be the case that for some $j\in\set{1,2}$, we have $f_{S_j}|_U \neq f_{S}|_U$. Therefore
\[
 \gamma \leq \Pr_{S_1,S_2 \supset U}\left[f_{S_1}|_U \neq f_{S_2}|_U\right] \leq 2\Pr_{\substack{T \supset U \\ S_j,S \supset T}}\left[f_{S_j}|_U \neq f_{S}|_U\right] \leq 2\Pr_{\substack{T \supset U \\ S_j,S \supset T}}\left[f_{S_j}|_T \neq f_{S}|_T\right].
\]
This shows that $\EE_U\left[\gamma\right] \leq 2\epsilon$.\end{proof}

\begin{claim}[bounding $\delta$]$\EE_U[\delta] = O(\epsilon)$.\label{clm:dim1rub}\end{claim}
\begin{proof}
Let $Y = \{ S \supset U \mid f_S|_U \neq g|_U \}$, so that $\Pr\left[S \in Y\right] = \gamma$. Let $B$ be the set of bad $i$. If $i$ is bad then
\[
 \Pr_{S \supset U}\left[S \in Y \mid i \in S\right] \geq \frac{1}{2}.
\]
Hence the reverse union bound shows that
\[
 \delta = \Pr_{S \supset U}\left[ S \text{ is bad}\right] = \Pr_{S \supset U}\left[S \text{ hits } B\right] = O\left(\Pr_{S \supset U}\left[S \in Y\right]\right) = O(\gamma).
\]
Hence $\EE_U\left[\delta\right] = O(\EE_U\left[\gamma\right]) = O(\epsilon)$.\end{proof}

\begin{claim}[bounding $\gamma_i$] For all $i \notin U$, $\EE_{U,i \notin U}\left[\gamma_i\right]= O(\eps/t)$.\label{clm:dim1gammai}\end{claim}
\begin{proof}
  We first consider the case when $i \notin U$ is good.
\begin{align*}
 \gamma_i &= \Pr_{S \supset U+i}\left[f_S|_U = g|_U\right] \cdot
            \Pr_{\substack{S \supset U+i \\ f_S|_U = g|_U}}\left[f_S(i) \neq g(i)\right] \\
          &\leq
 \Pr_{S \supset U+i}\left[f_S|_U = g|_U\right] \cdot
 \Pr_{\substack{S_1,S_2 \supset U+i \\ f_{S_1}|_U = f_{S_2}|_U = g|_U}}\left[f_{S_1}(i) \neq f_{S_2}(i)\right] & [\text{plurality vs. collision}]\\ &=
 \left(\Pr_{S \supset U+i}\left[f_S|_U = g|_U\right]\right)^{-1} \cdot
 \Pr_{S_1,S_2 \supset U+i}\left[f_{S_1}|_U = f_{S_2}|_U = g|_U \text{ and } f_{S_1}(i) \neq f_{S_2}(i)\right] \\ &\leq
 2\Pr_{S_1,S_2 \supset U+i}\left[f_{S_1}|_U = f_{S_2}|_U \text{ and } f_{S_1}(i) \neq f_{S_2}(i)\right],
\end{align*}
where we have used the fact that $i$ is good at the very last step. Clearly, the same bound also holds for bad $i$.

This shows that
\begin{align*}
 \EE_{U,i \notin U}\left[\gamma_i\right] &\leq
 2\EE_T \Pr_{\substack{S_1,S_2 \supset T \\ i \in T}}\left[f_{S_1}|_T \text{ and }f_{S_2}|_T \text{ disagree only on } i\right]\\&=
 \frac{2}{t}\EE_T \Pr_{S_1,S_2 \supset T}\left[f_{S_1}|_T \text{ and }f_{S_2}|_T \text{ disagree on a single element}\right] \\ &\leq
 \frac{2}{t}\EE_T \Pr_{S_1,S_2 \supset T}\left[f_{S_1}|_T \text{ and }f_{S_2}|_T \text{ disagree}\right] \\ &\leq \frac{2}{t} \epsilon. \qedhere
\end{align*}
\end{proof}

Combining the three bounds, we conclude that
\[
 \EE_{U,S \supset U}\left[f_S \neq g_U|_S\right] \leq O(\epsilon) + k' \frac{2}{t} \epsilon = O\left(\frac{\epsilon}{\zeta}\right).
\]
This completes the proof of \cref{lem:agree1-local}.
\end{proof}

\subsubsection{Part 2: sewing}

In this step, we show that  the global function $g_U$
corresponding to some $(t-1)$-set $U$ explains most local functions
$f_S$ corresponding to $S$'s not necessarily containing $U$. We will
first prove this under the assumption that $k \geq 4t$ (or equivalently $\xi \geq \nicefrac34$), and then extend it
to all $\xi \in (0,1)$.

\begin{lemma}\label{lem:agree1-global}
For all positive integers $n, k, t$ satisfying $1 \leq t < k \leq n$ and alphabet
$\Sigma$ the following holds.
Let $\{f_S\colon S \to \Sigma \mid S \in \binom{[n]}{k}\}$ be an ensemble
of local functions satisfying
\begin{equation}\label{eq:ag1-global}
 \Pr_{(T,S_1,S_2) \sim \nu_n(k,t)}\left[f_{S_1}|_{T} \neq f_{S_2}|_{T}\right] \leq  \epsilon,
\end{equation}
then there exists a global function $g\colon [n] \to \Sigma$ such that $\Pr_S\left[f_S \neq g|_S\right] = O(\eps/\xi\zeta)$, where $\zeta = t/(k-t+1)$ and $\xi = 1-t/k$.
\end{lemma}

\begin{proof}

  Let us first assume that $k \geq 4t$. \cref{lem:agree1-local} implies that
 \[\Pr_{\substack{U_1,U_2 \\ S \supset U_1 \cup U_2}}\left[g_{U_1}|_S \neq g_{U_2}|_S\right] = O(\epsilon/\zeta).\]

We now define a coupling $(U_1,U_2,S,S_1,S_2)$ as follows: Choose two random $(t-1)$-sets $U_1,U_2$. Then, choose a random $k$-set $S$ containing $U_1$.
Choose a random ordering of $S \setminus U_1$, let $X_1$ be the first $\lfloor (k-t+1)/2 \rfloor$ elements, and let $X_2$ be the subsequent $\lceil (k-t+1)/2 \rceil$ elements. Finally, choose a random $k$-set $S_1$ containing $U_1 \cup U_2 \cup X_1$, and a random $k$-set $S_2$ containing $U_1 \cup U_2 \cup X_2$. The last step is possible since $2(t-1)+\lceil (k-t+1)/2 \rceil \leq 2t+k/2 \leq k$. The marginal distribution of $(S_j,U_1,U_2)$ is random sets $U_1,U_2$ and a random set $S_j \supset U_1 \cup U_2$. The coupling ensures that $S \subset S_1 \cup S_2$, and so $g_{U_1}|_S \neq g_{U_2}|_S$ implies that $g_{U_1}|_{S_j} \neq g_{U_2}|_{S_j}$ for some $j \in \{1,2\}$. Therefore
\[
 \Pr_{\substack{U_1,U_2 \\ S \supset U_1}}\left[g_{U_1}|_S \neq g_{U_2}|_S\right] = O(\epsilon/\zeta).
\]
Fix $U := U_2$ such that
\[
 \Pr_{U_1, S \supset U_1}\left[g_{U_1}|_S \neq g_U|_S\right] = O(\epsilon/\zeta).
\]
Then $g = g_U$ satisfies
\[
 \Pr_{U_1, S \supset U_1}\left[f_S \neq g|_S\right] \leq
 \Pr_{U_1, S \supset U_1}\left[f_S \neq g_{U_1}|_S\right] +
 \Pr_{U_1, S \supset U_1}\left[g_{U_1}|_S \neq g|_S\right] = O(\epsilon/\zeta).
\]
This completes the case of small $t$.

Suppose now that $t \geq k/4$, and let $t' = \lfloor k/4 \rfloor$. For any $0< \ell \leq t$, define  \[\eps_\ell:= \Pr_{(T,S_1,S_2) \sim \nu_n(k,\ell)}\left[f_{S_1}|_{T} \neq f_{S_2}|_{T}\right].\] Given the hypothesis \eqref{eq:ag1-global} that $\eps_t \leq \eps$, we will show that $\eps_{t'} = O(\eps/\xi)$. We can then apply the preceding case and complete the proof.

Let $r < t$ such that $r \geq 2t-k$. Consider the following coupling argument. Consider the joint distribution $(R,T_1,T_2,S_1,S_2,S)$ obtained as follows: Choose a random $r$-set $R$. Choose two random $t$-sets $T_1,T_2 \supset R$. Choose a random $k$-set $S_1$ containing $T_1$, a random $k$-set $S_2$ containing $T_2$, and a random $k$-set $S$ containing $T_1 \cup T_2$. The last step is possible if $k \geq 2t-r$, i.e., if $r \geq 2t-k$.

The sets $S_1,S_2$ are two random sets containing $R$, and the sets $S,S_j$ are two random sets containing $T_j$. Furthermore, if $f_{S_1}|_R \neq f_{S_2}|_R$ then $f_{S_j}|_{T_j} \neq f_S|_{T_j}$ for some $j \in \{1,2\}$. Therefore
\[
 \Pr_{R, S_1,S_2 \supset R}\left[f_{S_1}|_R \neq f_{S_2}|_R\right] \leq \Pr_{\substack{T_1\\ S_1,S \supset T_1}}\left[f_{S_1}|_{T_1} \neq f_{S}|_{T_1}\right] + \Pr_{\substack{T_2\\S_2,S \supset T_2}}\left[f_{S_2}|_{T_2} \neq f_{S}|_{T_2}\right] \leq 2\epsilon.
\]
This demonstrates that if the hypothesis for the agreement theorem is
true for a particular choice of $n, k,t$ (i.e., $\eps_t \leq \eps$), then the hypothesis is also
true for $n,k,r$ by increasing $\epsilon$ to $2\epsilon$ (i.e., $\eps_r \leq 2\eps$) provided
$r \geq 2t-k$ or equivalently $2(1-t/k) \geq (1-r/k)$. Thus, given the hypothesis is true for some $t,k$ satisfying $\xi=1-t/k$, we can perform the above coupling argument $\lceil \log_2 \nicefrac1\xi \rceil$ times to reduce $t$ to less than $k/4$. Each step doubles the error, and so the overall error probability becomes $O(\epsilon/\xi)$. Now applying the argument for small $t$, we obtain that there exists a global $g$ such that $\Pr_S\left[f_S \neq g|_U\right] = O(\epsilon/\xi\zeta)$.
\end{proof}

\subsubsection{Part 3: majority decoding}

In this final step of the argument, we show that \cref{lem:agree1-global} can be further strenghtened by showing that the global function $g$ can in fact be the majority-decoded function.

\begin{proof}[Proof of \cref{thm:agree1}]
Let $g$ be the global function given by \cref{lem:agree1-global}.
Define $G \fdom{[n]}$ as follows: let $G(i)$ be a plurality value of $f_S(i)$ over all $S \ni i$ (breaking ties arbitrarily). Let $B = \{ i \mid g(i) \neq G(i) \}$, and let $Y = \{ S \mid f_S \neq g|_S \}$.

For each $i \in B$,
\[
 \Pr_S\left[S \in Y \mid i \in S\right] \geq \Pr_S\left[f_S(i) \neq g(i) \mid i \in S\right] \geq \frac{1}{2},
\]
since otherwise $g(i)$ would be the unique majority value. Applying the reverse union bound, we conclude that
\[
 \Pr_S\left[S \text{ hits } B\right] = O\left(\Pr_S\left[f_S \neq g|_S\right]\right) = O(\epsilon/\zeta\xi).
\]
Therefore
\[
 \Pr_S\left[f_S \neq G|_S\right] \leq
 \Pr_S\left[f_S \neq g|_S\right] + \Pr_S\left[g|_S \neq G|_S\right] =
 \Pr_S\left[f_S \neq g|_S\right] + \Pr_S\left[S \text{ hits } B\right] = O(\epsilon/\zeta\xi).
\]
\cref{thm:agree1} follows by observing that $1/\zeta\xi = O\left(\frac{1}{\min(\frac{t}{k},1-\frac{t}{k})}\right)$.
\end{proof}

\section{Agreement theorem for high dimensions}\label{sec:agreed}

\begin{theorem}[agreement theorem for high dimensions]\label{thm:agreed}
For every positive integer $d$ there exists a constant $M_d = 2^{O(d^4)}$ such that for all positive integers $n, k, t$ satisfying $d \leq t < k \leq n$ and alphabet
$\Sigma$ the following holds.
Let $\{f_S\colon \bind{S} \to \Sigma \mid S \in \binom{[n]}{k}\}$ be an ensemble
of local functions satisfying
\[
 \Pr_{(T,S_1,S_2) \sim \nu_n(k,t)}\left[f_{S_1}|_{T} \neq f_{S_2}|_{T}\right] \leq  \epsilon,
\]
then the majority decoded function $G\colon \bnd \to \Sigma$ satisfies
\[ \Pr_{S\sim \nu_n(k)}\left[f_S \neq G|_S\right] \leq \frac{M_d}{\xi\zeta^{d}} \cdot \eps\,,\] where $\zeta = \min\{1,t/(k-t+1)\}$ and $\xi=1-t/k$, and the majority-decoded function $G$ is the one given by ``popular vote'', namely for each $A \in \bnd$ set $G(A)$ to be the most frequently occurring value among $\set{f_S(A) \mid S\supset A}$ (breaking ties arbitrarily).
\end{theorem}

The same two examples showing the tightness of \cref{thm:agree1} also show that an inverse polynomial dependence on $\xi$ and $\zeta$ are needed in the above theorem. However. it is unclear if the exponential dependence on $d$ (in both the $\zeta^d$ term and the constant $M_d$) is required.

Let $t+1 \leq k \leq n$ be as in the hypothesis.
We will assume the following convention throughout this section.
\begin{itemize}
	\item $S$ will always denote a set of size $k$.
	\item $T$ will always denote a set of size $t$.
	\item $U$ will always denote a set of size $t-d$.
	\item $n' = n-(t-d)$, $k' = k-(t-d)$, $t' = t-(t-d) = d$.
        \end{itemize}

As in the 1-dimensional case, we will first (in part 1) define for each $U$ a global function $g_U\fdom{\bnd}$ such that $\Pr_{U, S \supseteq U}\left[f_S \neq g_U|_S\right] = O(\epsilon)$, then show (in part 2) that $g_U$ corresponding to a typical $U$ explains most $S$ (not necessarily ones that contain $U$) and finally (in part 3) show that the majority-decoded function works.

\subsection{Part 1: relative decoding}

\newcommand{\bset}{L}

In this part, for each $(t-d)$-set $U$, we define a function $g_U\colon \binom{[n]}{d} \to \Sigma$ such that
\[
 \Pr_{U, S \supseteq U}\left[f_S \neq g_U|_S\right] = O(\epsilon).
\]
For a large part of the proof, we will focus on a single $U$, and use $g$ for $g_U$. The definition of such a $g=g_U$ will be inspired by the corresponding definition in the 1-dimensional case. We will define $g$ incrementally. We will first define $g$ for all $d$-sets fully contained in $U$, then for $d$-sets which have at most 1 element outside $U$, then for sets which have at most 2 elements outside $U$, and so on.

First for some notation. Let $h \fdom{\bnd}$. For any $0 \leq \ell \leq d$ and any $\ell$-set $\bset  \subset [n]\setminus U$, we define the function $h^{U}[\bset ] \fdom{\binom{U}{d-\ell}}$ as follows:
\[
 h^{U}[\bset ](A) := h(A\cup \bset ).
\]

For any $k$-set $S$ such that $S \supset U \cup \bset $, we can similarly define $f_S^{U}[\bset ] \fdom{\binom{U}{d-\ell}}$.
We will say that $f_S^{U} \lagree{\ell} h^{U}$ if $f^{U}_S[\bset ] = h^{U}[\bset ]$ whenever $|\bset | \leq \ell$. If the set $U$ is clear from context, we will drop the superscript $U$ from the above notation (as in $f_S[L]$ and $f_s \lagree{\ell} h$).

\begin{lemma}\label{lem:agreed-local}
For every positive integer $d$, there exists a positive constant $\kappa_d= 2^{O(d^4)}$ such that the following is true. For all positive integers $n, k, t$ satisfying $d \leq t < k \leq n$ and alphabet
$\Sigma$ the following holds.
Let $\{f_S\colon \bind{S} \to \Sigma \mid S \in \binom{[n]}{k}\}$ be an ensemble
of local functions satisfying
\[
 \Pr_{(T,S_1,S_2) \sim \nu_n(k,t)}\left[f_{S_1}|_{T} \neq f_{S_2}|_{T}\right] \leq  \epsilon,
\]
then there exists an ensemble $\{g_U\colon \bnd \to \Sigma \mid U \in
\binom{[n]}{t-d}\}$ of global functions such that when a random $U \in
\binom{[n]}{t-d}$ and $S \in \binom{[n]}{k}$ are
chosen satisfying $S \supset U$, then
\[ \Pr_{U, S \supset U}\left[g_U|_S \neq f_S\right] \leq
\frac{\kappa_d}{\zeta^{d}}\cdot \eps,\] where $\zeta = \min\{1,t/(k-t+1)\}$.
\end{lemma}
\begin{proof}
We begin with the definition of $g_U$. Fix some $(t-d)$-set $U$. We will define $g=g_U$ by defining $g[\bset ]$ incrementally for $|\bset |= 0\ldots d$. To this end, we need the notions of $\ell$-good and $\ell$-excellent sets for $-1 \leq \ell \leq d$. To begin with, all $k$-sets $S$ are $(-1)$-excellent. For larger $\ell$, we define inductively as follows:
\begin{itemize}
\item For $\ell = 0$ to $d$:
  \begin{itemize}
  \item For all $\ell$-sets $\bset $,
    \[ g[\bset ] := \operatorname{popular}\{f_S[\bset ] \mid S \supset U \cup \bset , S \text{ is $(\ell-1)$-excellent}\}. \]
    If there are no $(\ell-1)$-excellent $k$-sets $S$, we define $g[\bset ]$ arbitrarily.
    \item An $\ell$-set $\bset $ is said to be {\em good}  if it satisfies the following condition (and said to be {\em bad} otherwise).
      \[
 \Pr_{S \supset U}\left[\text{$S$ is $(\ell-1)$-excellent} \mid S \supseteq \bset \right] \geq \frac{1}{2}.
\]
\item A $k$-set $S$ is said to be {\em $\ell$-good} if it is $(\ell-1)$-excellent and does not contain any bad $\ell$-set $\bset $.
\item A $k$-set $S$ is said to be {\em $\ell$-excellent} if it is $\ell$-good and $f_S \lagree{\ell} g$.
\end{itemize}
\end{itemize}

If $S$ is $\ell$-excellent then $f_S \lagree{\ell} g$. In particular, if $S$ is $d$-excellent then $f_S\lagree{d} g$ which implies that $f_S = g|_S$. We thus have,
\begin{align}
  \Pr_{S \supset U}\left[ f_S \neq g|_S\right] & \leq \Pr_{S \supset U}\left[ S \text{ is not $d$-excellent}\right]\label{eq:notdexcel}
\end{align}
Towards bounding this probability, we prove the following claims. For $\ell = 0,\dots,d$, let
\begin{align*}
  \gamma_\ell & := \Pr_{S \supset U}\left[ S \text{ is $\ell$-good but not $\ell$-excellent}\right]
\end{align*}

\begin{claim}\label{clm:smallint} If $n,k,t,d, \ell$ are 5 positive integers satisfying $0 \leq \ell \leq d \leq t <k \leq n$, then
  \[
    \Pr_{\substack{R\colon |R| = t-\ell\\ S_1, S_2 \supset R}}\left[ f_{S_1}|_R \neq f_{S_2}|_R\right] \leq 2^{\ell} \cdot \Pr_{\substack{T\colon |T| = t\\ S_1, S_2 \supset T}}\left[ f_{S_1}|_T \neq f_{S_2}|_T\right] .
  \]
\end{claim}
\begin{proof}
  The proof of this claim employs a coupling argument identical to the one used in the proof of \cref{dlm:dim1gamma}. Sample $(d+1)$ $k$-sets $S_{-1},S_0,\ldots, S_d$ using the following process.
  \begin{enumerate}
  \item Choose $t+1$ distinct random elements $a_0, a_1,\ldots, a_t$.
  \item For $\ell \in \set{0,1,\ldots,d}$
    \begin{itemize}
    \item Choose a random element $a'_\ell$ distinct from the $t$-set $\set{ a'_i \mid 0\leq i < \ell } \cup \set{a_i \mid \ell< i \leq t}$.
    \end{itemize}
  \item Choose a random $(k-t-1)$-set $W_{-1}$ disjoint from the $(t+1)$-set $\set{a_0,\ldots, a_t}$.
  \item Set $S_{-1} \gets \set{a_0, \ldots, a_t} \cup W_{-1}$.
  \item For $\ell \in \set{0,1,\ldots, d}$
    \begin{itemize}
    \item Choose a random $(k-t-1)$-set $W_\ell$ disjoint from the $(t+1)$-set $\set{a'_i \mid i \leq \ell } \cup \set{a_i \mid i > \ell }$.
    \item Set $S_\ell \gets \set{ a'_i \mid i \leq \ell } \cup \set{a_i \mid i > \ell } \cup W_\ell$.
    \item Set $T_\ell \gets \set{ a'_i \mid i < \ell } \cup \set{a_i \mid i > \ell }$.
    \item Set $R_\ell \gets \set{a_i \mid i > \ell }$.
    \end{itemize}
  \end{enumerate}
 Such a coupling is feasible as long as $k \geq t+1$. It follows from the definition that for any $\ell \in \set{0,1,\ldots, d}$, the triple $(T_\ell, S_{\ell-1},S_\ell)$ is distributed according to $\nu_n(k,t)$ while the triple $(R_\ell,S_{-1},S_\ell)$ is distributed according to $\nu_n(k,t-\ell)$. Furthermore, observe that $R_\ell \subseteq T_i$ for all $i \leq \ell$. Hence, if  $f_{S_{-1}}|_{R_\ell} \neq f_{S_\ell}|_{R_\ell}$, it must be the case that for some $i\in\set{1,2,\ldots,\ell}$, we have $f_{S_{i-1}}|_{T_i} \neq f_{S_i}|_{T_i}$. Hence,
\[
    \Pr\left[ f_{S_{-1}}|_{R_\ell} \neq f_{S_\ell}|_{R_\ell}\right] \leq \sum_{i=0}^\ell \Pr\left[ f_{S_{i-1}}|_{T_i} \neq f_{S_i}|_{T_i}\right] = (\ell+1) \Pr_{(T,S,S') \sim \nu_n(k,t)}\left[ f_{S}|_{T} \neq f_{S'}|_{T}\right] \,. \qedhere
  \]
\end{proof}

\begin{claim} For every positive integer $d$, there exists a positive constant  $C'_d= 2^{O(d)}$ such that  for integers $\ell, t, k, n$ satisfying $0\leq \ell \leq d \leq t < k\leq n$ the following holds.
  \[ \EE_U[\gamma_\ell] = \EE_U\left[\Pr_{S \supset U}\left[ S \text{ is $\ell$-good but not $\ell$-excellent}\right] \right] \leq \frac{C'_d\cdot \eps}{\zeta^\ell}.
  \]
\end{claim}
\begin{proof}
  The proof of this claim is, in some sense, a generalization of the proof of \cref{clm:dim1gammai}. To begin with, let us fix a $(t-d)$-set $U$.
\begin{align*}
\gamma_\ell &= \Pr_{S \supset U}\left[ S \text{ is $\ell$-good but not $\ell$-excellent}\right]\\
         &= \Pr_{S \supset U}\left[ S \text{ is $\ell$-good and } \exists \ell\text{-set } \bset \subset \overline{U}\text{ such that } S \supset \bset\text{ and }f_S[\bset] \neq g[\bset]\right] & \text{[Definition of $\ell$-excellent set]}\\
         & \leq \sum_{\bset: |\bset| = \ell, L \subset \overline{U}} \Pr_{S \supset U}\left[ S \supset \bset \text{ and } S \text{ is $\ell$-good  and } f_S[\bset] \neq g[\bset]\right] & \text{[Union Bound]}\\
         & \leq \sum_{\substack{\bset: |\bset| = \ell\\\bset \text{ good}, L \subset \overline{U}}}\Pr_{S \supset U}\left[S \supset \bset \text{ and } S \text{ is $\ell$-good  and } f_S[\bset] \neq g[\bset]\right] &[S \text{ is $\ell$-good } \&\, \bset \subset S \Rightarrow \bset \text{ is good}]\\
         & \leq \sum_{\substack{\bset: |\bset| = \ell\\\bset \text{ good}, L \subset \overline{U}}}\underbrace{\Pr_{S \supset U}\left[S \supset \bset \text{ and } S \text{ is $(\ell-1)$-excellent  and } f_S[\bset] \neq g[\bset]\right]}_{\gamma_\ell(\bset)}&[S \text{ is $\ell$-good } \Rightarrow S \text{ is $(\ell-1)$-excellent}]\\
  \end{align*}
  We now turn to bounding each term $\gamma_\ell(\bset)$ in the above sum for good $L$.
  \begin{align*}
    \gamma_\ell(\bset) &:= \Pr_{S \supset U}\left[S \supset \bset \text{ and } S \text{ is $(\ell-1)$-excellent  and } f_S[\bset] \neq g[\bset]\right]\\
                &=  \Pr_{S \supset U}\left[S \supset \bset\right] \cdot \Pr_{S\supset U \cup \bset}\left[S \text{ is $(\ell-1)$-excellent}\right]\cdot\Pr_{S \supset U \cup \bset}\left[f_S[\bset] \neq g[\bset] \mid S \text{ is $(\ell-1)$-excellent}\right]\\
                & \leq \Pr_{S \supset U}\left[S \supset \bset\right] \cdot \Pr_{S\supset U \cup \bset}\left[S \text{ is $(\ell-1)$-excellent}\right]\cdot\Pr_{S_1,S_2 \supset U \cup \bset}\left[f_{S_1}[\bset] \neq f_{S_2}[\bset] \mid S_1, S_2 \text{ are $(\ell-1)$-excellent}\right]\\
                & = \Pr_{S \supset U}\left[S \supset \bset\right] \cdot \frac{\Pr_{S_1,S_2 \supset U \cup \bset}\left[f_{S_1}[\bset] \neq f_{S_2}[\bset] \text{ and } S_1, S_2 \text{ are $(\ell-1)$-excellent}\right]}{\Pr_{S\supset U \cup \bset}\left[S \text{ is $(\ell-1)$-excellent}\right]}\\
    & \leq 2 \Pr_{S \supset U}\left[S \supset \bset\right] \cdot \Pr_{S_1,S_2 \supset U \cup \bset}\left[f_{S_1}[\bset] \neq f_{S_2}[\bset] \text{ and }S_1, S_2 \text{ are $(\ell-1)$-excellent}\right],
\end{align*}
where in the last step we have used the fact that $\bset$ is good. Plugging this back into the earlier expression, we have
\begin{align*}
  \gamma_\ell \leq \sum_{\substack{\bset: |\bset| = \ell\\\bset \text{ good}, L \subset\overline{U}}}\gamma_\ell(\bset) &\leq 2 \sum_{\substack{\bset: |\bset| = \ell\\\bset \text{ good}, L \subset \overline{U}}}\Pr_{S \supset U}\left[S \supset \bset\right] \cdot \Pr_{S_1,S_2 \supset U \cup \bset}\left[f_{S_1}[\bset] \neq f_{S_2}[\bset] \text{ and }S_1, S_2 \text{ are $(\ell-1)$-excellent}\right]\\
           &\leq 2 \sum_{\bset: |\bset| = \ell, L \subset\overline{U}}\Pr_{S \supset U}\left[S \supset \bset\right] \cdot \Pr_{S_1,S_2 \supset U \cup \bset}\left[f_{S_1}[\bset] \neq f_{S_2}[\bset] \text{ and } S_1, S_2 \text{ are $(\ell-1)$-excellent}\right]\\                                                                                            & = 2 \binom{k'}{\ell} \EE_{\bset: |\bset| = \ell, L \subset \overline{U}}\left[\Pr_{S_1,S_2 \supset U \cup \bset}\left[f_{S_1}[\bset] \neq f_{S_2}[\bset] \text{ and } S_1, S_2 \text{ are $(\ell-1)$-excellent}\right]\right]\\                                                                                              & \leq 2 \binom{k'}{\ell} \EE_{\bset: |\bset| = \ell, L \subset \overline{U}}\left[\Pr_{S_1,S_2 \supset U \cup \bset}\left[f_{S_1}[\bset] \neq f_{S_2}[\bset] \text{ and } f_{S_1} \lagree{\ell-1}f_{S_2} \right]\right] \, .
\end{align*}
We now average over $U$:
\begin{align*}
  \EE_U\left[\gamma_\ell\right] & \leq 2 \binom{k'}{\ell} \EE_{U, \bset \in \binom{\overline{U}}{\ell}}\left[\Pr_{S_1,S_2 \supset U \cup \bset}\left[f_{S_1}^U[\bset] \neq f_{S_2}^U[\bset] \text{ and } f_{S_1}^U\lagree{\ell-1}f_{S_2}^U \right]\right]\\
                  & \leq 2 \binom{k'}{\ell} \EE_{\substack{R : |R| = t-d+\ell\\\bset \in \binom{R}{\ell}}}\left[\Pr_{S_1,S_2 \supset R}\left[f_{S_1}|_R \neq f_{S_2}|_R \text{ and } f_{S_1}^{R\setminus \bset}\lagree{\ell-1}f_{S_2}^{R\setminus \bset} \right]\right]\\
  &= 2 \binom{k'}{\ell} \Pr_{\substack{R : |R| = t-d+\ell\\S_1, S_2 \supset R}}\left[f_{S_1}|_R \neq f_{S_2}|_R \right] \cdot \EE_{\substack{R : |R| = t-d+\ell\\S_1, S_2 \supset R}}\left[\Pr_{\bset \in \binom{R}{\ell}}\left[ f_{S_1}^{R\setminus \bset}\lagree{\ell-1}f_{S_2}^{R\setminus \bset} \middle| f_{S_1}|_R \neq f_{S_2}|_R\right]\right]
\end{align*}
The first probability expression in the final expression above is upper bounded by $(\ell+1)\eps$. The second probability expression is bounded as follows: Consider any $R, S_1, S_2 \supset R$ such that $f_{S_1}|_R \neq f_{S_2}|_R$. In particular, let $A\subset R$ be any $d$-set such that $f_{S_1}(A) \neq f_{S_2}(A)$. Now, if $f_{S_1}^{R\setminus \bset}\lagree{\ell-1}f_{S_2}^{R\setminus \bset}$, we must have that $\bset \subset A$. In other words, the second probability expression is bounded above $\Pr_{\bset : \bset \supset R, |\bset| = \ell}\left[\bset \subset A\right] = \binom{d}{\ell}/\binom{t-d+\ell}{\ell}$. We thus, have
\[\EE_U\left[\gamma_\ell\right] \leq \frac{2^{\ell}\binom{k'}{\ell}\cdot \binom{d}{\ell}}{\binom{t-d+\ell}{\ell}} \cdot O(\eps) = \frac{2^{O(d)}}{\zeta^\ell} \cdot \eps\,. \qedhere\]
\end{proof}

\begin{claim}For every positive integer $d$, there is a positive constant $C_d = 2^{O(d^3)}$ such that for all $\ell = 0,\dots,d$, we have
  \[\Pr_{S \supset U}\left[ S \text{ is not $\ell$-good}\right] \leq  C_d\cdot \Pr_{S \supset U}\left[ S \text{ is not $(\ell-1)$-excellent}\right].\] \end{claim}
\begin{proof}
This claim is proved using the Reverse Union Bound as in the proof of \cref{clm:dim1rub}.
Let $C_d := \max_{\ell}\{c'(\ell,\nicefrac12)\}+1= 2^{O(d^3)}$, where $c'(d,c)$ is the constant guaranteed by the Reverse Union Bound, \cref{lem:rub}.
Define $Y \subset \binom{[n]}{k}$ and $B \subset \binom{[n]}{\ell}$ as follows:
  \begin{align*}
    Y &:= \left\{S \in \binom{[n]}{k} \mid S \text{ in not $(\ell-1)$-excellent}\right\}\,,\\
    B &:= \left\{\bset \in \binom{[n]}{\ell} \mid \bset \text{ is bad}\right\}\,.
  \end{align*}
  By definition of bad set, we have that
  \[\Pr_{S \supset U}\left[S \in Y  \mid S \supseteq \bset \right] \geq \frac{1}{2}.\]
  Hence, by reverse union bound,
  \begin{align*}
    \Pr_{S \supset U}\left[ S \text{ is not $\ell$-good}\right] &\leq \Pr[ S \text{ hits } B] + \Pr _{S \supset U}[S \text{ is not $(\ell-1)$-excellent}]\\
    & \leq (C_d-1)\cdot \Pr _{S \supset U}[S \text{ is not $(\ell-1)$-excellent}]+ \Pr _{S \supset U}[S \text{ is not $(\ell-1)$-excellent}].\qedhere
  \end{align*}
  \end{proof}

We now return to our goal of bounding $\Pr_{U,S \supset U}\left[ f_S \neq g|_S\right]$ via the inequality \eqref{eq:notdexcel}:
\begin{align*}
  \EE_U\Pr_{S \supset U}\left[ f_S \neq g|_S\right] & \leq \EE_U\Pr_{S \supset U}\left[ S \text{ is not $d$-excellent}\right]\\                                                    &\leq \EE_U\Pr_{S \supset U}\left[ S \text{ is not $d$-good}\right] + \EE_U\Pr_{S \supset U}\left[ S \text{ is $d$-good but not $d$-excellent}\right] \\                                                    &\leq C_d\cdot \EE_U \Pr_{S \supset U}\left[ S \text{ is not $(d-1)$-excellent}\right] + \frac{C'_d\cdot \eps}{\zeta^d}\\                                                    & \leq C_d^2\cdot \EE_U \Pr_{S \supset U}\left[ S \text{ is not $(d-2)$-excellent}\right] +C_d \cdot \frac{C'_d\cdot \eps}{\zeta^{d-1}}+ \frac{C'_d\cdot \eps}{\zeta^d}\\                                                    &\leq C_d^\ell \cdot \EE_U\Pr_{S \supset U}\left[ S \text{ is not $(d-\ell)$-excellent}\right] + \sum_{i=1}^\ell C_d^{\ell-i} \cdot \frac{C'_d\cdot \eps}{\zeta^{d-(\ell-i)}}\\                                                    &\leq C_d^{d+1}\cdot \EE_U\Pr_{S \supset U}\left[ S \text{ is not $(-1)$-excellent}\right] + \sum_{i=1}^{d+1} C_d^{d+1-i} \cdot \frac{C'_d\cdot \eps}{\zeta^{i-1}}\\
  &= \frac{2^{O(d^4)}\cdot \eps}{\zeta^{d}}. \qedhere
\end{align*}
\end{proof}

\subsection{Part 2: sewing}

In this step, we show that  the global function $g_U$
corresponding to some $(t-d)$-set $U$ explains most local functions
$f_S$ corresponding to $S$'s not necessarily containing $U$. We will
first prove this under the assumption that $k \geq 2(d+1)t$ (or equivalently $\xi \geq 1-\nicefrac1{2(d+1)}$) and then extend it
to all $\xi \in (0,1)$.

\begin{lemma}\label{lem:agreed-global}
For every positive integer $d$, let $\kappa_d= 2^{O(d^4)}$ be the positive constant from \cref{lem:agreed-local}. For all positive integers $n, k, t$ satisfying $d \leq t < k \leq n$ and alphabet
$\Sigma$ the following holds.
Let $\{f_S\colon \bind{S} \to \Sigma \mid S \in \binom{[n]}{k}\}$ be an ensemble
of local functions satisfying
\begin{equation}\label{eq:agd-global}
 \Pr_{(T,S_1,S_2) \sim \nu_n(k,t)}\left[f_{S_1}|_{T} \neq f_{S_2}|_{T}\right] \leq  \epsilon,
\end{equation}
then there exist a global function $g \colon \bnd \to \Sigma$ such that
\[ \Pr_{S}\left[f_S \neq g|_S\right] \leq
  \frac{2(2d+4)\kappa_d}{\xi\zeta^{d}}\cdot \eps\, ,\] where $\zeta = \min\{1,t/(k-t+1)\}$ and $\xi=1-t/k$.
\end{lemma}

\begin{proof}

  Let us first assume that $k \geq 2(d+1)t$. \cref{lem:agreed-local} implies that
 \[\Pr_{\substack{U_1,U_2 \\ S \supset U_1 \cup U_2}}\left[g_{U_1}|_S \neq g_{U_2}|_S\right] \leq  \frac{2\kappa_d}{\zeta^{d}}\cdot \eps.\]

 We now define a coupling $(U_1,U_2,S,S_1,S_2)$ as follows: Choose two random $(t-d)$-sets $U_1,U_2$. Then, choose a random $k$-set $S$ containing $U_1$. Randomly partition $S \setminus U_1$ into $d+1$ sets $R_1,\ldots,R_{d+1}$ of (almost) equal size $\frac{k-t+d}{d+1}$. Finally, for each $j \in [d+1]$, choose a random $k$-set $S_j$ containing $U_1 \cup U_2 \cup [(S \setminus U_1) \setminus R_j]$. The last step is possible since $2(t-d) + \frac{d}{d+1} (k-t+d) \leq k$. The marginal distribution of $(S_j,U_1,U_2)$ is that of two random sets $U_1,U_2$ and a random set $S_j \supset U_1 \cup U_2$. The coupling ensures that $\binom{S}{d}$ is covered by $\bigcup_j \binom{S_j}{d}$ (since for any $d$ elements in $S$, there must be some $R_j$ missing). Hence,
 $g_{U_1}|_S \neq g_{U_2}|_S$ implies that $g_{U_1}|_{S_j} \neq g_{U_2}|_{S_j}$ for some $j \in \{1,\dots,d+1\}$. Therefore
\[
 \Pr_{\substack{U_1,U_2 \\ S \supset U_1}}\left[g_{U_1}|_S \neq g_{U_2}|_S\right] = \sum_{j=1}^{d+1} \Pr_{\substack{U_1,U_2 \\ S_j \supseteq U_1 \cup U_2}}[g_{U_1}|_{S_j} \neq g_{U_2}|_{S_j}]  \leq \frac{2(d+1)\kappa_d}{\zeta^{d}}\cdot \eps.
\]
Fix $U := U_2$ such that
\[
 \Pr_{U_1, S \supset U_1}\left[g_{U_1}|_S \neq g_U|_S\right] \leq \frac{2(d+1)\kappa_d}{\zeta^{d}}\cdot \eps.
\]
Then $g = g_U$ satisfies
\begin{align*}
 \Pr_{U_1, S \supset U_1}\left[f_S \neq g|_S\right] &\leq
 \Pr_{U_1, S \supset U_1}\left[f_S \neq g_{U_1}|_S\right] +
 \Pr_{U_1, S \supset U_1}\left[g_{U_1}|_S \neq g|_S\right]\\
& \leq  \frac{2\kappa_d}{\zeta^{d}}\cdot \eps + \frac{2(d+1)\kappa_d}{\zeta^{d}}\cdot \eps = \frac{(2d+4)\kappa_d}{\zeta^{d}}\cdot \eps.
\end{align*}
This completes the case of small $t$.

Suppose now that $t \geq k/2(d+1)$, and let $t' = \lfloor k/2(d+1) \rfloor$. This case is handled identically to the corresponding case in dimension one.
For any $0< \ell \leq t$, define  \[\eps_\ell:= \Pr_{(T,S_1,S_2) \sim \nu_n(k,\ell)}\left[f_{S_1}|_{T} \neq f_{S_2}|_{T}\right].\]
Mirroring the argument for dimension one, we obtain that if the hypothesis for the agreement theorem is
true for a particular choice of $n, k,t$ (i.e., $\eps_t \leq \eps$), then the hypothesis is also
true for $n,k,r$ by increasing $\epsilon$ to $2\epsilon$ (i.e., $\eps_r \leq 2\eps$) provided
$r \geq 2t-k$, or equivalently $2(1-t/k) \geq (1-r/k)$. Thus, given the hypothesis is true for some $t,k$ satisfying $\xi=1-t/k$, we can perform this argument $m:=\lceil \log_2 \nicefrac1\xi\cdot(1-\nicefrac1{2(d+1)})\rceil $ times to reduce $t$ to less than $k/2(d+1)$. Each step doubles the error, and so the overall error probability becomes $2^m\cdot \eps \leq \nicefrac2\xi(1-\nicefrac1{2(d+1)})\cdot \eps$. Now applying the argument for small $t$, we obtain that there exists a global $g$ such that
\[\Pr_S\left[f_S \neq g|_U\right] \leq  \frac{(2d+4)\kappa_d}{\zeta^{d}}\cdot \frac2\xi\left(1-\frac1{2(d+1)}\right)\cdot \eps \leq \frac{2(2d+4)\kappa_d}{\xi\zeta^{d}}\cdot \eps\, .\qedhere\]
\end{proof}

\subsection{Part 3: majority-decoding}

The proof is identical to the dimension one case.

\begin{proof}[Proof of \cref{thm:agreed}]
Let $g$ be the global function given by \cref{lem:agreed-global} that satisfies
\[\Pr_S\left[f_S \neq g|_S\right]  \leq \frac{2(2d+4)\kappa_d}{\xi\zeta^{d}}\cdot \eps =: \eta \,.\]

Define $G \fdom{\bnd}$ as follows: for each $A \in \bnd$ set $G(A)$ to be the most frequently occurring value among $\set{f_S(A) \mid S\supset A}$ (breaking ties arbitrarily).

Let $B = \left\{ A \in \bnd \mid g(A) \neq G(A) \right\}$, and let $Y = \left\{ S \mid f_S \neq g|_S \right\}$.

For each $A \in B$,
\[
 \Pr_S\left[S \in Y \mid A \subset S\right] \geq \Pr_S\left[f_S(A) \neq g(A) \mid A \subset S\right] \geq \frac{1}{2},
\]
since otherwise $g(A)$ would be the unique majority value. Applying the Reverse Union Bound~\cref{lem:rub}, we conclude that
\[
 \Pr_S\left[S \text{ hits } B\right] \leq  C'(d,\nicefrac12)\cdot \Pr_S\left[f_S \neq g|_S\right]  \leq C'(d,\nicefrac12)\cdot \eta.
\]
Therefore
\[
 \Pr_S\left[f_S \neq G|_S\right] \leq
 \Pr_S\left[f_S \neq g|_S\right] + \Pr_S\left[g|_S \neq G|_S\right] =
 \Pr_S\left[f_S \neq g|_S\right] + \Pr_S\left[S \text{ hits } B\right] = (C'(d,\nicefrac12)+1)\cdot \eta.
\]
\cref{thm:agreed} follows.
\end{proof}

\section{More agreement theorems} \label{sec:agreement-other}

In this section, we present several variants of the $d$-dimensional agreement theorem \cref{thm:agreed}.

The first variant will be the simplest variant in which the local function is of the form
\[f_S\colon \binld{S} \to \Sigma\] I.e., it specifies a value for all $\ell$-subsets $A$ of $S$ such that $\ell \leq d$ instead of just $d$-subsets of $S$.
This variation is obtained by a simple union bound over \cref{thm:agreed}, one for each integer $\ell \in [0,d]$.

\begin{theorem}\label{thm:agreed-dless}
For every positive integer $d$, let $M_d = 2^{O(d^4)}$ be the positive constant from \cref{thm:agreed}. For all positive integers $n, k, t$ satisfying $d \leq t < k \leq n$ and alphabet
$\Sigma$ the following holds.
Let $\{f_S\colon \binld{S} \to \Sigma \mid S \in \binom{[n]}{k}\}$ be an ensemble
of local functions satisfying
\[
 \Pr_{(T,S_1,S_2) \sim \nu_n(k,t)}\left[f_{S_1}|_{T} \neq f_{S_2}|_{T}\right] \leq  \epsilon,
\]
then the majority decoded function $G\colon \binld{[n]} \to \Sigma$ satisfies
\[ \Pr_{S\sim \nu_n(k)}\left[f_S \neq G|_S\right] \leq \frac{(d+1)M_d}{\xi\zeta^{d}} \cdot \eps\,,\] where $\zeta = \min\{1,t/(k-t+1)\}$ and $\xi=1-t/k$, and the majority-decoded function $G$ is the one given by ``popular vote'', namely for each $A \in \binld{[n]}$ set $G(A)$ to be the most frequently occurring value among $\set{f_S(A) \mid S\supset A}$ (breaking ties arbitrarily).
\end{theorem}

\subsection{The exact intersection setting}
The $d$-dimensional agreement theorems (\cref{thm:agreed} and \cref{thm:agreed-dless}) are both proved when the triple of sets $(S_1,S_2,T)$ are picked according to the distribution $\nu_n(k,t)$. We now extend the agreement to the exact intersection setting when the triple of sets $(S_1,S_2,T)$ are instead picked according to the distribution $\mu_n(k,t)$ as follows:
\begin{itemize}
\item Pick a random $t$-set $T \subseteq [n]$.
\item Pick two $k$-sets $S_1$ and $S_2$ uniformly from the set of pairs of $k$-sets that have exact intersection $T$. In other words,
\[ (S_1,S_2) \gets_u \left\{ (S'_1,S'_2) \in \binom{n}{k} \times \binom{n}{k} \middle| S'_1 \cap S'_2 = T \right\}.\]
\end{itemize}
Note that for the distribution $\mu_n(k,t)$ to have non-zero support, we must have $t+2(k-t) = 2k-t \leq n$.
\begin{theorem}[exact-intersection setting (strenghtenend form of \cref{thm:agreed-intro})]\label{thm:agreed-exact}
For every positive integer $d$, let $M_d = 2^{O(d^4)}$ be the positive constant from \cref{thm:agreed}. Let $\rho \in (0,1)$. For all positive integers $n, k, t$ satisfying $d \leq t < k \leq n$ and $2k-t \leq (1-\rho)n$ and any alphabet
$\Sigma$ the following holds.
Let $\{f_S\colon \binld{S} \to \Sigma \mid S \in \binom{[n]}{k}\}$ be an ensemble
of local functions satisfying
\[
 \Pr_{(T,S_1,S_2) \sim \mu_n(k,t)}\left[f_{S_1}|_{T} \neq f_{S_2}|_{T}\right] \leq  \epsilon,
\]
then the majority decoded function $G\colon \binld{[n]} \to \Sigma$ satisfies
\[ \Pr_{S\sim \nu_n(k)}\left[f_S \neq G|_S\right] \leq (1/\rho)^{O(1)} \cdot \frac{(d+1)M_d}{\xi\zeta^{d}} \cdot \eps\,,\] where $\zeta = \min\{1,t/(k-t+1)\}$ and $\xi=1-t/k$, and the majority-decoded function $G$ is the one given by ``popular vote'', namely for each $A \in \binld{[n]}$ set $G(A)$ to be the most frequently occurring value among $\set{f_S(A) \mid S\supset A}$ (breaking ties arbitrarily).
\end{theorem}
\begin{proof}
For every $0 \leq r \leq k-t$, we define the following intermediate distributions $\nu_n(k,t,r)$ on triples of sets as follows.
\begin{itemize}
	\item Choose a random $t$-set $T$.
	\item Choose a disjoint random $r$-set $R$.
	\item Choose two random $k$-sets $S_1,S_2 \subseteq T \cup R$ whose intersection is exactly $T \cup R$.
	\item Output $(S_1,S_2,T)$.
\end{itemize}
Observe that $\mu_n(k,t) = \nu_n(k,t,0)$.

Given an ensemble $\calF=\{f_S\colon \binld{S} \to \Sigma \mid S \in \binom{[n]}{k}\}$ of local functions, define for each $r \in [0,k-t]$,
\begin{equation}\label{eq:rtest} \eps_{n,k,t}^r(\calF) : = \Pr_{(T,S_1,S_2) \sim \nu_n(k,t,r)}\left[f_{S_1}|_{T} \neq f_{S_2}|_{T}\right].
\end{equation}
We are given that $\eps_{n,k,t}^0(\calF) \leq \eps$.
Applying \cref{claim:rtest} $t$ many times, we get that $\eps_{n,k,t}^r(\calF) \leq 3^t \eps$ for all $r \leq (2^t-1) \rho n$, as a simple induction shows. We are interested in a bound that works for all $r \leq k-t$, and for that it suffices to take $t_0 = O(\log \frac{k-t}{\rho n}) = O(\log (1/\rho))$. Therefore $\eps_{n,k,t}^r(\calF) \leq 3^{t_0} \eps = (1/\rho)^{O(1)} \eps$.
The theorem follows since $\nu_n(k,t)$ is a convex combination of $\nu_n(k,t,r)$ for $r \leq k-t$.
\end{proof}

\begin{claim}\label{claim:rtest}
Suppose $2k-t \leq (1-\rho)n$. If $t+r+\Delta \leq k$ and $\Delta-r \leq \rho n$, then $\eps_{n,k,t}^{r+\Delta}(\calF) \leq 3\cdot \eps_{n,k,t}^{r}(\calF)$.
\end{claim}
\begin{proof}
We prove this using the following coupling argument:
\begin{itemize}
\item Choose disjoint sets $T$ of size $t$, $R$ of size $r$, $D_1,D_2,D_3$ of size $\Delta$, and $A,B$ of size $k-t-r-\Delta$.
\item Set $S_1 \gets T \cup R \cup D_1 \cup A$.
\item Let $S_2 \gets T \cup R \cup D_2 \cup B$.
\item Let $S_3 \gets T \cup R \cup D_3 \cup A$.
\item Let $S_4 \gets T \cup R \cup D_1 \cup B$.
\end{itemize}
We can choose such sets if $t+r+3\Delta + 2(k-t-r-\Delta) \leq n$, which translates to $2k-t + (\Delta-r) \leq n$ which is true since $2k-t \leq (1-\rho)n$ and $\Delta- r \leq \rho n$.

All the sets contain $T$. The sets $S_1,S_4$ are random supersets of $T$ whose intersection is exactly $t+r+\Delta$. Thus, the marginal of the triple $(S_1,S_4,T)$ is exactly $\nu_n(k,t,r+\Delta)$.
Every two adjacent sets $S_i, S_{i+1}$ ($i \in \{1,2,3\})$ are random supersets of $T$ whose intersection is exactly $t+r$. Thus, the marginal of the triple $(S_i,S_{i+1},T)$ is exactly $\nu_n(k,t,r)$. Clearly, if $f_{S_1}|_{T} \neq f_{S_4}|_{T}$, then it must be the case that for some $i \in \{1,2,3\}$ we have $f_{S_i}|_{T} \neq f_{S_{i+1}}|_{T}$. Hence, $\eps_{n,k,t}^{r+\Delta}(\calF) \leq 3\cdot \eps_{n,k,t}^{r}(\calF)$.
\end{proof}

\subsection{The $p$-biased distribution}

The $d$-dimensional agreement theorems (\cref{thm:agreed,thm:agreed-dless,thm:agreed-exact}) are all proved with respect to the distribution $\nu_{n}(k)$, the uniform distribution over $k$-sized subsets of $[n]$. In this section, we extend these results to the $p$-biased setting $\mu_p$. In this setting, the distribution $\nu_{n}(k,t)$ is replaced by the
distribution $\mu_{p,\paramq}$, which is a distribution over pairs $S_1,S_2$
of subsets of $[n]$ defined as follows. For each element $x$
independently, we put $x$ only in $S_1$ or only in $S_2$ with
probability $p(1-\paramq)$ (each), and we put $x$ in both with probability
$p\paramq$. This is possible if $p(2-\paramq) \leq 1$. Note that if sets $S_1, S_2$ are picked according to the distribution $\mu_{p,\paramq}$ then the marginal
distribution of each of $S_1$ and $S_2$ is $\mu_p$, while the distribution of the intersection $S_1 \cap S_2$ is $\mu_{p\paramq}$. We present below the $p$-biased version of \cref{thm:agreed-exact}.

\begin{theorem}[$p$-biased setting]\label{thm:agreed-pbiased}
For every positive integer $d$, let $M_d = 2^{O(d^4)}$ be the positive constant from \cref{thm:agreed} and $\rho \in (0,1)$.
For all positive numbers $p, \paramq \in (0,1)$ satisfying $p(2-\paramq) \leq 1-\rho$, positive integer $n$ and alphabet
$\Sigma$ the following holds.
Let $\{f_S\colon \binld{S} \to \Sigma \mid S \subseteq [n] \}$ be an ensemble
of local functions satisfying
\[
 \Pr_{(S_1,S_2) \sim \mu_{p,\paramq}}\left[f_{S_1}|_{S_1 \cap S_2} \neq f_{S_2}|_{S_1\cap S_2}\right] \leq  \epsilon,
\]
then the majority decoded function $G\colon \binld{[n]} \to \Sigma$ satisfies
\[ \Pr_{S\sim \mu_p}\left[f_S \neq G|_S\right] \leq (1/\rho)^{O(1)} \cdot \frac{3(d+1)M_d}{\min\{(1-\paramq), \paramq^{d}\}} \cdot \eps\,,\] where the majority-decoded function $G$ is the one given by ``popular vote'', namely for each $A \in \binld{[n]}$ set $G(A)$ to be the most frequently occurring value among $\set{f_S(A) \mid S\supset A}$ (breaking ties arbitrarily).
\end{theorem}

\begin{proof}[Proof of \cref{thm:agreed-pbiased}]
 Let $N$ be a large integer and define $K = \lfloor pN \rfloor$, $T = \lfloor p\paramq N \rfloor$. For every $S \in \binom{[N]}{K}$, define $\tilde{f}_S := f_{S \cap [n]}$. In other words, for all $A \subset S \in \binom{[N]}{K}, |A| \leq d$, let $\tilde{f}_S(A) := f_{S\cap[n]}(A \cap [n])$. If $S_1,S_2 \sim \mu_N(K,T)$ then the distribution of $S_1 \cap [n],S_2 \cap [n]$ is close to $\mu_{p,\paramq}$, and so for large enough $N$ we have
\[
 \Pr_{S_1,S_2 \sim \mu_N(K,T)}[\tilde{f}_{S_1}|_{S_1 \cap S_2} \neq \tilde{f}_{S_2}|_{S_1 \cap S_2}] \leq 2\epsilon.
\]
Hence, the ensemble of functions $\{\tilde{f}_S\}_{S \in \binom{[N]}{K}}$ satisfies the hypothesis of the agreement theorem (\cref{thm:agreed}) with $\epsilon$ replaced by $2\eps$. Hence, by \cref{thm:agreed}, if we define $\tilde{G}\colon \binom{[N]}{\leq d} \to \Sigma$ by plurality decoding then $\Pr_{S \sim \mu_N(K)}[\tilde{f}_S \neq \tilde{G}|_S] = O_d(\epsilon)$. Since $\tilde{f}_S$ depends only on $S \cap [n]$, there exists a function $\hat{G} \colon \binom{[n]}{\leq d} \to \Sigma$ such that $\tilde{G}(T) = \hat{G}(T \cap [n])$. Moreover, for large enough $N$ the distribution of $S \cap [n]$ approaches $\mu_p$, and so $\hat{G} = G$.\yfnote{Expanded on the issue of several most common values.}\footnote{There's a fine point here: there could be several most common values, only some of which could be strictly most common for $\tilde{f}$. However, any most common value of $f$ is almost most common $\tilde{f}$, and the proof for the uniform setting can be straightforwardly adapted to handle such values.} This completes the proof.
\end{proof}

\subsection{Product distributions (for large $p$)}

The multiplicative decay in the $d$-dimensional agreement theorem, \Cref{thm:agreed-pbiased}, deteriorates as $\paramq$ goes very close to 1. For the application to testing Reed--Muller codes, we need an agreement theorem that holds for all $\paramq \in (0,1)$ but for $p\ge p_0$ some fixed constant. The following agreement theorem over the product distribution $\mu_{p,p}$ shows that it is easy to prove a counterpart of  \Cref{thm:agreed-pbiased} if one is allowed a multiplicative decay of $p^{-d}$. Note that for $p \ge p_0$, this loss is not an issue for our applications (as we think of $d$ as a constant and $p \ge p_0$ is a constant). This is no longer true when $p=o(1)$, a regime in which we need the stronger result proved in \Cref{thm:agreed-pbiased}, which is independent of $p$ (but still depends on $d$).

\begin{theorem} \label{thm:agreement-alt}
For every positive integer $d$ and alphabet $\Sigma$, the following holds for all $p$. If $\{ f_S \colon \binom{S}{\leq d} \to \Sigma \mid S \in \{0,1\}^n \}$ is an ensemble of functions satisfying
\[
 \Pr_{S_1,S_2 \sim \mu_{p,p}}[f_{S_1}|_{S_1 \cap S_2} \neq f_{S_2}|_{S_1 \cap S_2}] = \epsilon
\]
then the global function $G\colon \binom{[n]}{\leq d} \to \Sigma$ defined by plurality decoding satisfies
\[
 \Pr_{S \sim \mu_p} [f_S \neq G|_S] \leq 2p^{-d}\epsilon.
\]
\end{theorem}
\begin{proof}
 The main observation behind the proof is this: if we choose $S_1,S_2 \sim \mu_p$ independently, then $(S_1,S_2) \sim \mu_{p,p}$.

 Consider now an arbitrary $f_{S_1}$. Suppose that $f_{S_1} \neq G|_{S_1}$, and choose an entry $T$ such that $f_{S_1}(T) \neq G(T)$. If we choose $S_2 \sim \mu_p$, then $T \subseteq S_2$ with probability $p^{|T|} \geq p^d$. Moreover, since $f_{S_1}(T) \neq G(T)$, the probability that $f_{S_2}(T) = f_{S_1}(T)$ is at most $1/2$. Therefore the probability that $f_{S_1}$ and $f_{S_2}$ disagree on their intersection is at least $p^d/2$. This shows that
\[
 \epsilon = \Pr_{S_1,S_2 \sim \mu_{p,p}}[f_{S_1}|_{S_1 \cap S_2} \neq f_{S_2}|_{S_1 \cap S_2}] \geq \frac{p^d}{2} \Pr_{S_1 \sim \mu_p} [f_{S_1} \neq G|_{S_1}]. \qedhere
\]
\end{proof}

\section{Testing Reed--Muller codes} \label{sec:reed-muller}

Every Boolean function $f\colon \bits^n \to \bits$ can be written in a unique way as $P \bmod 2$, where $P$ is a \emph{Boolean polynomial}, that is, a sum of distinct multilinear monomials. The \emph{Boolean degree} of $f$, denoted $\bdeg(f)$, is the degree of this polynomial.

The well-known BLR test~\cite{BlumLR1993,BellareCHKS1996} checks whether a given Boolean function has Boolean degree~1. Alon~\etal~\cite{AlonKKLR2005} developed the following $2^{d+1}$-query test $T_d$, which is a generalization of the BLR test to large Boolean degrees.\\

\begin{algorithm}[H]
    \setstretch{1.35}
  \SetKwInOut{Input}{Input}
  \Input{$f\colon \bits^n \to \bits$ (provided as an oracle)}
  \BlankLine
  Pick $x, a_1, \dots,a_{d+1} \in \bits^n$ independently from the distribution $\mu_{1/2}^{\otimes n}$, subject to the constraint that $a_1,\ldots,a_{d+1}$ are linearly independent.

  Accept iff \[\sum_{I \subseteq [d+1]} f\left(x + \sum_{i \in I}a_i\right) = 0 \pmod{2}\enspace. \]

  \caption{\textsc{AKKLR Low Degree Test $T_{d}$}}\label{alg:Td}
\end{algorithm}
\;

This test is closely related to the Gowers norms. An optimal analysis of the test was provided by Bhattacharyya~\etal~\cite{BhattacharyyaKSSZ2010}. We need a few definitions to state their result.

\begin{definition} \label{def:d-distance}
Let $f\colon \bits^n \to \bits$ and let $d \geq 0$. The \emph{distance to degree $d$} of $f$ is defined as follows:
\[
 \delta_d(f) := \min_{\bdeg(g) \leq d} \left\{\Pr_{x\sim \mu^{\otimes n}_{1/2}}[f(x) \neq g(x)] \right\} .
\]
\end{definition}

\begin{definition} \label{def:d-test}
Let $f\colon \bits^n \to \bits$ and let $d \geq 0$. The \emph{failure probability of test $T_d$} is
\[
 \rej_d(f) := \Pr_{\substack{x,a_1,\ldots,a_{d+1} \sim \mu_{1/2}^{\otimes n}\\a_1,\ldots,a_{d+1} \text{ linearly independent}}}\left[ \sum_{I \subseteq [d+1]} f\biggl(x + \sum_{i \in I} a_i\biggr) \neq 0 \pmod{2}\right].
\]
\end{definition}

\begin{remark}
 Another variant of $T_d$, which is closer to the definition of the Gowers norms, samples $a_1,\ldots,a_{d+1}$ without requiring them to be linearly independent. When $a_1,\ldots,a_{d+1}$ are linearly dependent, the test always succeeds. On the other hand, when $n \geq d+1$, the probability that $a_1,\ldots,a_{d+1}$ are linearly dependent is upper-bounded by a positive constant. Therefore when $n \geq d+1$, removing the constraint of linear independence only affects the rejection probability by a constant factor. Finally, when $n \leq d$ the test is pointless, since every function has Boolean degree at most~$d$.
\end{remark}

\begin{theorem}[\cite{BhattacharyyaKSSZ2010}] \label{thm:bkssz}
For every integer $d \geq 1$ there exists a constant $\epsilon_d > 0$ such that for all Boolean functions $f\colon \bits^n \to \bits$,
\[ \rej_d(f) \geq \min(2^d \delta_d(f), \epsilon_d). \]
\end{theorem}

\begin{corollary} \label{cor:bkssz}
For every integer $d \geq 1$ and all Boolean functions $f\colon \bits^n \to \bits$,
\[
 \delta_d(f) = O_d(\rej_d(f)).
\]	
\end{corollary}
\begin{proof}
 If $2^d \delta_d(f) \leq \epsilon_d$ then $\delta_d(f) \leq 2^{-d} \rej_d(f)$. Else, $\rej_d(f) \geq \epsilon_d$, and so $\delta_d(f) \leq 1 \leq \epsilon_d^{-1} \rej_d(f)$.	
\end{proof}

Our goal in this section is to extend the analysis of Bhattacharyya~\etal~to the $\mu_p$ setting (wherein we measure closeness of $f$ to degree~$d$ with respect to the $\mu_p$ measure instead of the $\mu_{1/2}$ measure). More precisely:
\begin{definition} \label{def:d-distance-p}
Let $f\colon \bits^n \to \bits$ and let $d \geq 0$. The \emph{$\mu_p$-distance to degree $d$} of $f$ is defined as follows:
\[
 \delta^{(p)}_d(f) := \min_{\bdeg(g) \leq d} \left\{\Pr_{x\sim \mu_p^{\otimes n}}[f(x) \neq g(x)] \right\}.
\]
\end{definition}
To this end, we consider the natural extension $T_{p,d}$ (\cref{alg:Tpd}) of the AKKLR test $T_d$ to the $\mu_p$ measure (restated below from the introduction). \\

\tpd*
\;

Observe that the points $(x+\sum_{i\in I} a_i)$, when viewed as points in $\bits^n$ (by filling the coordinates outside $S$ with $0$'s), are distributed individually according to $\mu_p^{\otimes n}$.

Let $\rej_{T}(f)$ denote the rejection probability of a test $T$ on input function $f$. In the rest of this section, we prove \cref{thm:bkssz-p} (restated below from the introduction).

\pBKSSZ*

\begin{remark}
The case for $p\in (\nicefrac12,1)$ is proved by reversing the roles of 0 and 1. In particular, by running the test $T_{p,d}$ on the function $\overline{f} =  (x_1,\ldots,x_n) \mapsto f(1-x_1,\ldots,1-x_n)$.
\end{remark}

The main ingredients that go into proving \cref{thm:bkssz-p} are the following.
\begin{description}
\item[$d$-dimensional agreement theorem:] We will use the $d$-dimensional agreeement theorems (\cref{thm:agreed-pbiased,thm:agreement-alt}) to reduce the analysis of the test $T_{p,d}$ in the $p$-biased setting to the analysis of the standard test $T_d$. Given $p,\paramq \in (0,1)$ such that $p(2-\paramq) \leq 1$ and $\kappa \geq 1$, we will say that \emph{agreement $\kappa$-holds for $(p,\paramq)$} if the hypothesis
\[ \Pr_{(S_1,S_2) \sim \mu_{p,\paramq}} [ P_{S_1}|_{S_1\cap S_2} \neq P_{S_2}|_{S_1\cap S_2} ] \leq \eps \]
implies the conclusion that there exists a global function $P\colon \binom{[n]}{\leq d} \to \{0,1\}$ such that
\[ \Pr_{S \sim \mu_p} [P_S \equiv P|_S ] \leq \kappa\cdot \eps.\]
Stated in this language, \cref{thm:agreed-pbiased} implies that ``agreement $\kappa_{d,\paramq,\rho}$-holds for $(p,\paramq)$ whenever $p(2-\paramq) \leq 1-\rho$'', where $\kappa_{d,\paramq}$ is some constant dependending on $d,\paramq,\rho$ (but independent of $p$ and $n$). Similarly, \cref{thm:agreement-alt} implies ``agreement $2p^{-d}$-holds for $(p,p)$''.
\item[Biased Schwartz--Zippel Lemma:] The standard Schwartz--Zippel Lemma states that if $P$ is a non-zero Boolean polynomial with $\bdeg(P) \leq d$, then $\Pr_{x \sim \mu_{\nicefrac12}^{\otimes n}}[P(x) \bmod 2 = 1] \geq (\frac12)^d$. The following is an extension of this lemma to the $p$-biased setting.
\begin{claim} \label{lem:biased-schwartz-zippel}
Suppose that $P$ is a non-zero polynomial with $\bdeg(f) \leq d$. Then for all $\theta \in (0,1)$,
\[
 \Pr_{\mu_\theta}[P \bmod 2 = 1] \geq \min(\theta,1-\theta)^d =:\delta^{(d)}_\theta.
\]
\end{claim}
\begin{proof}
Let $I$ be an inclusion-maximal set such that $P$ contains the monomial $\prod_{i \in I} x_i$. For every setting $\sigma$ of the variables outside $I$, we obtain a nonzero polynomial $P|_\sigma$ (since the monomial $\prod_{i \in I} x_i$ necessarily survives). Therefore there is at least one input on which $P|_\sigma \bmod 2 = 1$. Since $P|_\sigma$ is a polynomial on at most $d$ variables, this input occurs with probability at least $\min(\theta,1-\theta)^d$ in $\mu_\theta$.
\end{proof}
\item[BKSSZ Theorem for $p=\nicefrac12$:]
We will say that a test $T$ is \emph{$\rho$-valid for $p$} if the following holds.
\begin{itemize}
\item Completeness: $\rej_T(f) = 0$ whenever $\bdeg(f) \leq d$ and
\item Soundness: $\delta^{(p)}_d(f) \leq \rho\cdot \rej_T(f)$.
\end{itemize} \Cref{cor:bkssz} states that $T_d$ (modified so that it always accepts when the dimension is at most $d$) is $\rho_d$-valid for~$1/2$ for some constant $\rho_d$ dependent on $d$ (but independent of $n$).
\end{description}
The following lemma shows how we can use the agreement theorem to reduce the analysis of the test $T_{p,d}$ to that of $T_d$.
\begin{lemma} \label{lem:bkssz-analysis}
Suppose that $T_d$ is $\rho_d$-valid for $\nicefrac12$ and agreement $\kappa$-holds for $(2p,\paramq)$ for some $p\in (0,\nicefrac12)$ and $\paramq \in (\nicefrac12,1)$ satisfying $2p(2-\paramq) \leq 1$, then $T_{p,d}$ is $\left(\frac{4\kappa}{\delta^{(d)}_{\nicefrac1{2\paramq}}}+1\right)\cdot \rho_d$-valid for for $p$, where $\delta^{(d)}_{\theta}$ is as defined in \cref{lem:biased-schwartz-zippel}.
\end{lemma}
\begin{proof}
Let us start by noticing that completeness is clear, since $\bdeg(f|_S) \leq \bdeg(f)$. It remains to prove soundness.

By construction, $\rTpd(f) = \EE_{S \sim \mu_{2p}} [\rTd(f\restr_S)]$. The assumption that $T_d$ is $\rho_d$-valid for $\nicefrac12$ guarantees for each set $S \subseteq [n]$ the existence of a Boolean polynomial $P_S$ over $\{x_i \colon  i \in S\}$ of degree at most $d$ satisfying
\[\delta_S := \Pr_{x\sim \mu_{\nicefrac12}^{\otimes S}}\left[f\restr_S(x) \neq P_S(x) \bmod{2}\right] \leq \rho_d\cdot \rTd(f\restr_S). \]
Hence, $\delta:= \EE_{S \sim \mu_{2p}}[\delta_S] \leq \rho_d \cdot \rTpd(f)$.

Our goal now is to use agreement in order to sew the various polynomials $P_S$ into a global polynomial $P$. To this end, we will view each polynomial $P_S$ as a list of coefficients. In other words, $P_S\colon \binom{S}{\leq d} \to \{0,1\}$. To apply the agreement theorem, we need to show that the polynomials $P_S$ agree with each other. We show this using the $p$-biased variant of the Schwartz--Zippel Lemma, \cref{lem:biased-schwartz-zippel}.

Let $\theta = \nicefrac1{2\paramq}$. Note that $\theta \in (\nicefrac12,1)$ since $\paramq \in (\nicefrac12,1)$.
For two sets $S \supseteq T$, the polynomial $P_S|_T$ is the polynomial obtained by zeroing out all the coefficients of monomials involving variables outside $T$. Note that this matches with the definition $P_S|_T\colon \binom{T}{\leq d} \to \{0,1\}$ when $P_S$ is viewed as a list of coefficients, $P_S\colon \binom{S}{\leq d} \to \{0,1\}$. We define
\[
 \delta_{S,T} := \Pr_{x \sim \mu_\theta^{\otimes T}}[f\restr_T(x) \neq P_S|_T(x) \bmod{2}].
\]
Since $\paramq\theta =\nicefrac12$, we have
\begin{align*}
 \EE_{S \sim\mu_{2p}}\EE_{T \sim \mu_\paramq(S)}[\delta_{S,T}] &= \EE_{S \sim\mu_{2p}}\EE_{T \sim \mu_\paramq(S)}\Pr_{x\sim \mu_\theta^{\otimes T}}[f\restr_T(x) \neq P_S|_T(x) \bmod{2}]\\
&=
\EE_{S \sim\mu_{2p}}\Pr_{x\sim \mu_{\nicefrac12}^{\otimes S}} [f\restr_S(x) \neq P_S(x) \bmod{2}]
=\EE_{S \sim \mu_{2p}} [\delta_S]  = \delta.
\end{align*}
Hence, we have by Markov's inequality that \[\Pr_{S \sim \mu_{2p}, T \sim \mu_\paramq(S)}\left[\delta_{S,T} \geq  \frac{\delta^{(d)}_\theta}2\right] \leq \frac{2}{\delta^{(d)}_\theta}\cdot \delta.\]
By the union bound
\begin{align*}
\Pr_{(S_1,S_2) \sim \mu_{2p,\paramq}}\left[ \delta_{S_1,S_1 \cap S_2} \geq \frac{\delta^{(d)}_\theta}2 \text { or } \delta_{S_2,S_1 \cap S_2} \geq \frac{\delta^{(d)}_\theta}2 \right]
& \leq 2\cdot \Pr_{(S_1,S_2) \sim \mu_{2p,\paramq}} \left[ \delta_{S_1,S_1 \cap S_2} \geq \frac{\delta^{(d)}_\theta}2 \right]\\
& = 2\cdot \Pr_{S \sim \mu_{2p}, T \sim \mu_\paramq(S)} \left[ \delta_{S,T} \geq \frac{\delta^{(d)}_\theta}2 \right] \leq \frac{4}{\delta^{(d)}_\theta}\cdot \delta.
\end{align*}
If $\delta_{S_1,S_1 \cap S_2} < \nicefrac{\delta^{(d)}_\theta}2$ and $\delta_{S_2,S_1 \cap S_2} < \nicefrac{\delta^{(d)}_\theta}2$, then we have
\[\Pr_{x\sim \mu_\theta^{\otimes S_1 \cap S_2}}\left[ P_{S_1}|_{S_1\cap S_2} (x) \neq P_{S_2}|_{S_1\cap S_2}(x)\right] < \delta^{(d)}_\theta. \]
It then follows from the $p$-biased variant of the Schwartz--Zippel Lemma, \cref{lem:biased-schwartz-zippel}, that the two Boolean degree $d$ polynomials $P_{S_1}|_{S_1\cap S_2}$ and $P_{S_2}|_{S_1\cap S_2}$ must be identical. We thus have
\[ \Pr_{(S_1,S_2) \sim \mu_{2p,\paramq}}\left[ P_{S_1}|_{S_1\cap S_2} \neq  P_{S_2}|_{S_1\cap S_2} \right] \leq \frac{4}{\delta^{(d)}_\theta}\cdot \delta. \]
We have thus satisfied the hypothesis of the agreement theorem. Using the fact that ``agreement $\kappa$-holds for $(2p,\paramq)$'', we can conclude that there exists a global polynomial $P \colon \binom{[n]}{\leq d} \to \{0,1\}$ (and hence of Boolean degree at most $d$) such that
\[ \Pr_{S \sim \mu_{2p}}[P_S \neq P|_S] \leq  \frac{4\kappa}{\delta^{(d)}_\theta}\cdot \delta \,.\]
The following argument now shows that $P \bmod 2$ is close to $f$ under the $\mu_p$-measure, completing the proof of the lemma.
\begin{align*}
\Pr_{x \sim \mu_p^{\otimes n}}\left[f(x) \neq P (x) \bmod 2\right]
& = \EE_{S\sim \mu_{2p}} \left[ \Pr_{x \sim \mu_{\nicefrac12}^{\otimes S}}\left[ f\restr_S(x) \neq P|_S(x) \bmod 2\right]\right]\\
& \leq \EE_{S\sim \mu_{2p}} \left[ \Pr_{x \sim \mu_{\nicefrac12}^{\otimes S}}\left[ f\restr_S(x) \neq P_S(x) \bmod 2\right]\right] + \EE_{S\sim \mu_{2p}} \left[ \Pr_{x \sim \mu_{\nicefrac12}^{\otimes S}}\left[ P_S(x) \neq P|_S(x) \bmod 2\right]\right]\\
& \leq \EE_{S\sim \mu_{2p}} \left[\delta_S \right] + \Pr_{S\sim \mu_{2p}} \left[P_S \neq P|_S\right]\\
&\leq \delta +  \frac{4\kappa}{\delta^{(d)}_\theta}\cdot \delta \leq \left(1+\frac{4\kappa}{\delta^{(d)}_\theta}\right)\cdot \rho_d \cdot \rTpd(f)\, .\qedhere
\end{align*}
\end{proof}
An immediate application of \cref{lem:bkssz-analysis} using the agreement theorem \cref{thm:agreed-pbiased} is that $T_{p,d}$ is $C_{p,d}$-valid for $p$ for some constant $C_{p,d}$ (depending on $p$ and $d$). We will do a slightly more careful analysis to show that $T_{p,d}$ is $C_{d}$-valid for $p$ for some constant $C_{d}$ (depending on $d$ but not on $p$).
\begin{proof}[Proof of \cref{thm:bkssz-p}] Our starting point is \cref{cor:bkssz}, which states that $T_d$ is $\rho_d$-valid for~$1/2$. The proof then follows two cases: $p$ small and $p$ large.
\begin{description}
\item[$p \in (0,0.3)$:] We set $\paramq = \nicefrac34$. Observe that $2p(2-\paramq) = \nicefrac{5p}{2} < \nicefrac34 < 1$. \cref{thm:agreed-pbiased} implies that there exists a constant $\kappa_d$ such that for all $p \in (0,0.3)$, ``agreement $\kappa_d$-holds for $(2p,\paramq)$''. Applying \cref{lem:bkssz-analysis}, we get that $T_{p,d}$ is $C^{\text{small}}_d$-valid for all $p \in (0,0.3)$, where $C^{\text{small}}_d :=\left(\frac{4\kappa_d}{\delta^{(d)}_{\nicefrac23}}+1\right)\cdot \rho_d$.
\item[$p \in [0.3,0.5)$:] We set $\paramq = 2p$. Observe that $2p(2-\paramq) = 4p(1-p) <1$ and $\paramq \in [0.6,1)$. Hence, $\nicefrac1{2\paramq} \in (\nicefrac12,\nicefrac56]$ and $\delta^{(d)}_{\nicefrac1{2\paramq}} =\min(\nicefrac1{2\paramq},1-\nicefrac1{2\paramq})^d \geq (\nicefrac16)^d$. \cref{thm:agreement-alt} implies that there exists a constant $\kappa'_d:=2(0.3)^{-d}$ such that for all  $p \in [0.3,0.5)$, ``agreement $\kappa'_d$-holds for $(2p,\paramq)$''. Applying \cref{lem:bkssz-analysis}, we get that $T_{p,d}$ is $C^{\text{large}}_d$-valid for all $p \in [0.3,0.5)$ where $C^{\text{large}}_d := \left(\frac{4\kappa'_d}{\delta^{(d)}_{\nicefrac1{2\paramq}}}+1\right)\cdot \rho_d \leq \left(4\cdot 6^d\cdot \kappa'_d+1\right)\cdot \rho_d$.
\end{description}
Thus, $T_{p,d}$ is $C_d$-valid for all $p \in (0,0.5)$ where $C_d = \max\left(C^{\text{small}}_d, C^{\text{large}}_d\right)$.
\end{proof}

\section{Reverse union bound and hyper graph pruning lemma} \label{sec:rub}

In this section we prove \Cref{lem:rub}, which we now restate.

\rub*

We will deduce \Cref{lem:rub} from another result, which we call the \emph{hypergraph pruning lemma} (restated below from the introduction):

\hpl*

We use the following notation to prove \cref{lem:rub,lem:pruning-uniform}.
If $H \subseteq \binom{[n]}{d}$ and $S \subseteq [n]$ then
\[
 \hit H(S) = \{I \in H : I \subseteq S\}.
\]
We think of $\hit H(S)$ as a random variable that depends on the distribution of $S$. With this notation, the conclusion of \cref{lem:rub}  can be rephrased as follows.
\[
 \Pr_{S\in \nu_n(k)}[\hit B(S) \neq \emptyset] \leq C'(c,d) \cdot \Pr_{S \in \nu_n(k)}[S \in Y] \,,
\]
while the conclusion of \cref{lem:pruning-uniform} can be rewritten as
\begin{itemize}
\item $\Pr_{S\sim\nu_n(k)}[\hit {H'}(S) \neq \emptyset] \geq \Pr_{S\sim\nu_n(k)}[\hit H(S) \neq \emptyset] / C''(d,\epsilon)	$.
\item For every $I \in H'$,
\[
  \Pr_{S\sim\nu_n(k)}[\hit{H'}(S) = \{I\} \mid \hit{H'}(S) \ni I] \geq 1 - \epsilon.
\]
\end{itemize}
The hypergraph pruning lemma states that you can prune a given $d$-uniform hypergraph $H$ to a subhypergraph $H'$ for which the events ``$S \supseteq I$'' (where $S \sim \nu_n(k)$ and $I \in H'$) are nearly disjoint. For this to be nontrivial, we also require that the probability that $S$ contains \emph{some} hyperedge in $H'$ does not drop too much compared to $H$.

Let us now deduce \Cref{lem:rub} from \Cref{lem:pruning-uniform}.

\begin{proof}[Proof of \Cref{lem:rub}]
We apply \Cref{lem:pruning-uniform} to $H := B$ with $\epsilon := c/2$, obtaining a subset $B' \subseteq B$. For every $I \in B'$, we have
\begin{align*}
 \Pr_{S \sim \nu_n(k)}[S \in Y \text{ and } \hit{B'}(S) = \{I\}] &\geq
 \Pr_{S \sim \nu_n(k)}[S \supseteq I] \left(
 \Pr_{S \sim \nu_n(k)}[S \in Y \mid S \supseteq I] -
 \Pr_{S \sim \nu_n(k)}[\hit{B'}(S) \supsetneq \{I\} \mid S \supseteq I]  \right) \\ &\geq  \Pr_{S \sim \nu_n(k)}[S \supseteq I] \cdot  (c-\epsilon)= \Pr_{S \sim \nu_n(k)}[S \supseteq I]\cdot  \frac{c}{2}.
\end{align*}
The events ``$S \in Y \text{ and } \hit{B'}(S) = \{I\}$'' are disjoint for different $I$, and so
\[
 \Pr_{S \sim \nu_n(k)}[S \in Y] \geq
 \sum_{I \in B'} \Pr_{S \sim \nu_n(k)}[S \in Y \text{ and } \hit{B'}(S) = \{I\}] \geq
 \frac{c}{2} \sum_{I \in B'} \Pr_{S \sim \nu_n(k)}[S \supseteq I] \geq \frac{c}{2} \Pr_{S \sim \nu_n(k)}[\hit{B'}(S) \neq \emptyset].
\]
Together with $\Pr_{S\sim\nu_n(k)}[\hit{B'}(S) \neq \emptyset] \geq \Pr_{S\sim\nu_n(k)}[\hit{B}(S) \neq \emptyset]/C''(c,d)$, this completes the proof.
\end{proof}

It will be slightly simpler to present the proof of \Cref{lem:pruning-uniform} first in the product setting, that is, when $\nu_n(k)$ is replaced with the distribution $\mu_p$, which we do in \Cref{sec:rub-product}. We then explain the changes needed to adapt it to the uniform setting in \Cref{sec:rub-uniform}.

\subsection{Proof in product setting} \label{sec:rub-product}

The analog of \Cref{lem:pruning-uniform} in the product setting has exactly the same statement, but instead of choosing $S \sim \nu_n(k)$, we now choose $S \sim \mu_p([n])$.

Instead of proving the product version of \Cref{lem:pruning-uniform} directly, we will deduce it from a more detailed lemma which involves the concept of \emph{branching factor}.

\begin{definition}[branching factor] \label{def:bf}
	A $d$-uniform hypergraph $H$ has \emph{branching factor} $\rho$ if for all sets $A$ of size at most $d$, the number of hyperedges in $H$ extending $A$ is at most $\rho^{d-|A|}$.
	
	If this holds only for sets $A$ of size at least $d-k$, we say that $H$ has \emph{$k$-bounded branching factor $\rho$}.
\end{definition}

\begin{lemma} \label{lem:pruning-product-detailed}
For each $d \geq 1$ there exist constants $c_{\max}(d) = 2^{-O(d^2)}$ and $\kappa(d) = 2^{-O(d^2)}$ such that for every $c \leq c_{\max}(d)$ there exists a constant $K(d,c) = 2^{O_c(d^3)}$ such that for all positive $p \leq \kappa(d)c$, every $d$-uniform hypergraph $H$ has a subhypergraph $H'$ such that:
\begin{itemize}
\item Hitting property: $\Pr_{S\sim\mu_p}[\hit{H'}(S) \neq \emptyset] \geq \Pr_{S\sim\mu_p}[\hit H(S) \neq \emptyset]/K(d,c)$.
\item Branching factor property: $H'$ has branching factor $c/p$.
\end{itemize}
\end{lemma}

\Cref{lem:pruning-uniform} (in its product version) follows almost immediately from \Cref{lem:pruning-product-detailed} given the following lemma.

\begin{lemma} \label{lem:bf-hit}
If $c \leq 1$ and $H$ is a $d$-uniform hypergraph with branching factor $c/p$ then for all $I \in H$,
\[
 \Pr_{S\sim\mu_p}[\hit H(S) = \{I\} \mid \hit H(S) \ni I] \geq 1 - 2^dc.
\]
\end{lemma}
\begin{proof}
Recall that $\hit H(S)$ consists of all sets in $H$ that are contained in $S$. Conditioning on $\hit H(S) \ni I$ is the same as conditioning on $S \supseteq I$. If $\hit H(S) \supsetneq \{I\}$ then $S\supset J$ for some $J\in H$. For a fixed $J$, conditioned on $S\supseteq I$, this happens iff $S \setminus I \supseteq J \setminus I$ for some $J \in H$, which happens with probability $p^{|J \setminus I|} = p^{d-|J \cap I|}$. Therefore
\[
 \Pr_{S\sim\mu_p}[\hit H(S) \supsetneq \{I\} \mid \hit H(S)  \ni I] \leq \sum_{\substack{J \in H \\ J \neq I}} p^{|J \setminus I|} \leq \sum_{K \subsetneq I} \sum_{\substack{J \in H \\ J \cap I = K}} p^{d - |K|}.
\]
For each $K \subsetneq I$, the number of $J \in H$ extending $K$ by $d - |K|$ elements is at most $(c/p)^{d-|K|}$, and so
\[
 \Pr_{S\sim\mu_p}[\hit H(S) \supsetneq \{I\} \mid \hit H(S) \ni I] \leq \sum_{K \subsetneq I} \sum_{\substack{J \in H \\ J \cap I = K}} p^{d - |K|} \leq
 \sum_{K \subsetneq I} c^{d - |K|} < 2^d c,
\]
since there are $2^d-1$ summands, and each one is at most $c$ (since $c \leq 1$).
\end{proof}

We can now deduce the product version of \Cref{lem:pruning-uniform} from \Cref{lem:pruning-product-detailed}.

\begin{proof}[Proof of product version of \Cref{lem:pruning-uniform}]
Let $c = \min(\epsilon/2^d,c_{\max}(d),1)$. If 	$p \leq \kappa(d) c$ then the lemma follows directly from \Cref{lem:pruning-product-detailed} via \Cref{lem:bf-hit}, so suppose that $p > \kappa(d) c$. If $H$ is empty then we can take $H'$ to be empty as well, so suppose that $H$ contains some hyperedge $I$. We take $H' = \{I\}$. This works since
\[
 \Pr_{S\sim\mu_p}[\hit{H'}(S) \neq \emptyset] = p^d \geq (\kappa(d) c)^d \geq \Pr_{S\sim\mu_p}[\hit H(S) \neq \emptyset]/(1/\kappa(d) c)^d,
\]
and $1/\kappa(d) c = 2^{O(d^2)} \max(2^d/\epsilon,1/c_{\max}(d),1)$.
\end{proof}

In the sequel we will need the following simple corollary of \Cref{lem:bf-hit}.

\begin{corollary} \label{cor:bf-hit}
If $c \leq 1$ and $H$ is a $d$-uniform hypergraph with branching factor $c/p$ then
\[
 (1 - 2^dc) p^d |H| \leq \Pr_{S\sim\mu_p}[\hit H(S) \neq \emptyset] \leq p^d |H|.
\]	
\end{corollary}
\begin{proof}
The upper bound follows from the union bound, since the probability of hitting any particular set is $p^d$. The lower bound follows from
\[
 \Pr_{S\sim\mu_p}[\hit H(S)  \neq \emptyset] \geq \sum_{I \in H} \Pr_{S\sim\mu_p}[\hit H(S) = \{I\}] = p^d \sum_{I \in H} \Pr_{S\sim\mu_p}[\hit H(S) = \{I\} \mid \hit H(S) \ni I]
\]
via an application of \Cref{lem:bf-hit}.
\end{proof}

\medskip

The proof of \Cref{lem:pruning-product-detailed} is by induction. The case $d = 1$ is almost trivial, and the case $d = 2$ illustrates most of the ideas appearing in the general case. We therefore start by proving the lemma in these two cases, and only then present the general inductive argument.

Below, we will write $\Pr[\hit H \neq \emptyset]$ as a shortcut for $\Pr_{S\sim\mu_p}[\hit H(S) \neq \emptyset]$.

\subsubsection{Proof for \texorpdfstring{$d = 1$}{d=1}}

The idea is that there are two cases: either $H$ is \emph{sparse} (contains at most $c/p$ elements) or it is \emph{dense}. In the first case, take $H'=H$ and there is nothing to prove. In the second case, if we retain an arbitrary subset of $\lfloor c/p \rfloor$ elements (``pruning''), then a simple calculation shows that the hitting property is satisfied, essentially since the probability of hitting a set of $\Theta(1/p)$ elements is constant.

\begin{proof}[Proof of \Cref{lem:pruning-product-detailed} for $d = 1$]

We choose $c_{\max}(1) = 1/3$ and $\kappa(1) = 1$. If $H$ has branching factor $c/p$ then $H' := H$ satisfies the conditions of the lemma, so suppose that $|H| > c/p$. Let $H'$ be an arbitrary choice of $\lfloor c/p \rfloor$ elements of $H$. Note that $c/p \geq 1$ (since we assume that $p \leq c$), and so $\lfloor c/p \rfloor \geq c/(2p)$ (the worst case is when $c/p = 2-\delta$). Clearly $H'$ has branching factor $c/p$, and by \Cref{cor:bf-hit},
\[
 \Pr[\hit {H'} \neq \emptyset] \geq (1-2c) p |H| \geq \frac{1}{3} \cdot p \cdot \frac{c}{2p} = \frac{c}{6}. \qedhere
\]	
\end{proof}

This argument requires $c_{\max}(1) < 1/2$, in order for the lower bound in \Cref{cor:bf-hit} to be nontrivial. However, it is clear that we get a nontrivial bound even if $c \geq 1/2$. Nevertheless, we use \Cref{cor:bf-hit} to simplify calculations in more complicated situations.

\subsubsection{Proof for \texorpdfstring{$d = 2$}{d=2}}

When $d = 2$, the hypergraph $H$ is a graph. There are two ways for $H$ to fail to satisfy the branching factor condition. First, $H$ could have vertices of high degree (more than $c/p$). Second, it could simply contain too many edges (more than $(c/p)^2$).

Let us first give a sketch before the formal proof. For the first case, assuming the second case doesn't hold, there are at most $2c/p$ high degree vertices. For each one choose an arbitrary subset of $c/p$ edges touching it and remove the rest. The resulting graph will actually have degree up to $4c/p$, so we will have to dial down $c$.
We handle the second case by arbitrarily retaining $(c/p)^2$ edges. Since we can assume that there are no high-degree vertices, we can appeal to \Cref{cor:bf-hit} in order to lower-bound the hitting probability of the resulting graph.

\begin{proof}[Proof of \Cref{lem:pruning-product-detailed} for $d = 2$]
We choose $c_{\max}(2) = 1/5$ and $\kappa(2) = 1/2$.

Let $\gamma = c/2$, and note that $\gamma/p \geq 1$. Let $B$ be the set of vertices of degree at least $\gamma/p$, and let $K$ be the graph resulting from $H$ by removing all edges incident to a vertex in $B$. Clearly
\[
 \Pr[\hit H \neq \emptyset] \leq \Pr[\hit B \neq \emptyset] + \Pr[\hit K \neq \emptyset],
\]
and so one of the two latter probabilities is at least $\Pr[\hit H \neq \emptyset]/2$. We consider these two cases separately.

\paragraph{Case 1: $\Pr[\hit K \neq \emptyset] \geq \Pr[\hit H \neq \emptyset]/2$.} If $|K| \leq (\gamma/p)^2$ then we can take $H' := K$. Otherwise, let $H'$ be an arbitrary subset of $K$ consisting of $\lfloor (\gamma/p)^2 \rfloor \geq \gamma^2/(2p^2)$ edges. By construction, $H'$ has branching factor $\gamma/p \leq c/p$, and so \Cref{cor:bf-hit} shows that
\[
 \Pr[\hit K \neq \emptyset] \geq (1-4\gamma) p^2 \frac{\gamma^2}{2p^2} = \Omega(\gamma^2) \geq \Omega(\gamma^2) \Pr[\hit H \neq \emptyset],
\]
since $4\gamma \leq 2/5$.

\paragraph{Case 2: $\Pr[\hit B \neq \emptyset] \geq \Pr[\hit H \neq \emptyset]/2$.} We start by applying \Cref{lem:pruning-product-detailed} inductively on $B$ with the parameter $\gamma$; note $\gamma \leq c_{\max}(2)/2 \leq c_{\max}(1)$ and $p \leq c/2 \leq \gamma = \kappa(1) \gamma$. This results in a set $B' \subseteq B$ of size at most $\gamma/p$ such that
\[
 \Pr[\hit{B'} \neq \emptyset] \geq \Pr[\hit B \neq \emptyset]/K(1,\gamma) \geq \Pr[\hit H \neq \emptyset]/(2K(1,\gamma)).
\]

Each vertex $i \in B'$ had degree at least $\gamma/p$ in the original graph $H$. We choose a set $N_i$ of $\lfloor \gamma/p \rfloor \geq \gamma/(2p)$ edges adjacent to $i$ for each $i \in B'$, and form $H'$ by taking the union of the sets $N_i$ for all $i \in B'$. Since every edge is adjacent to at most two vertices, each edge appears in at most two sets $N_i$, and so $|H'| \geq \gamma/(4p) \cdot |B'|$.

The graph $H'$ has branching factor $c/p$. Indeed, the degree of each vertex $i \in B'$ is at most $N_i + |B'| \leq 2(\gamma/p) = c/p$, since each other $N_j$ contributes at most one neighbor of $i$. Similarly, the degree of any other vertex is at most $|B'| \leq c/p$. The total number of edges is at most $\gamma/p \cdot |B'| \leq (\gamma/p)^2 \leq (c/p)^2$. Therefore $H'$ has branching factor $c/p$.

Applying \Cref{cor:bf-hit} twice, we see that
\[
 \Pr[\hit{H'} \neq \emptyset] \stackrel{(1)}\geq (1-4c) p^2 |H'| \geq \frac{p^2}{5} \frac{c}{8p} |B'| \stackrel{(2)}\geq \frac{c}{40} \Pr[\hit{B'} \neq \emptyset] \geq \frac{c}{80K(1,\gamma)} \Pr[\hit H \neq \emptyset],
\]
completing the proof.
\end{proof}

\subsubsection{General case}

The general case is quite similar to the case $d = 2$. There are now $d$ ways for the hypergraph to fail to satisfy the branching factor property, and accordingly we decompose $H$ into $d$ parts (in fact, we use $d+1$ parts).

Let $B_1$ be all $(d-1)$-sets with high degree in $H$. Let $H_1$ be the result of removing from $H$ all extensions of sets in $B_1$.
Let $B_2$ be all $(d-2)$-sets with high degree in $H_1$, and let $H_2$ be the result of removing from $H_1$ all extensions of sets from $B_2$. We continue in this manner and consider $B_1,B_2,\ldots,B_d,H_d$. Since $\Pr[\hit H\neq\emptyset] \ge \Pr[\hit{B_1}\neq\emptyset]+\cdots +\Pr[\hit{B_d}\neq\emptyset]+\Pr[\hit{H_d}\neq\emptyset]$ at least one of these is non-negligible and we use it to inductively construct the hypergraph $H'$ with bounded branching factor.
\begin{proof}[Proof of \Cref{lem:pruning-product-detailed}]
The proof is by induction on $d$. In the base case $d = 0$ we can take $c_{\max}(0) = \kappa(0) = 1$ and $K(0,c) = 1$, since $H' = H$ always works.

Suppose now that $d \geq 1$. We choose
\begin{align*}
&c_{\max}(d) = \min(c_{\max}(0)/2^{d-1}, \ldots, c_{\max}(d-1)/2^{d-1},1/2^{d+1}), \\
&\kappa(d) = \min(\kappa(0),\ldots,\kappa(d-1))/2^{d-1}, \\
&K(d,c) = 2^{3d} \max(K(0,c/2^{d-1}), \ldots, K(d-1,c/2^{d-1}))/c.
\end{align*}
It is not hard to verify that $c_{\max}(d),\kappa(d) \geq 2^{-O(d^2)}$. As for $K(d,c)$, it is not hard to check that $K(d,c)$ is monotone in $d$, and so
\[
 K(d,c) = \frac{2^{3d}}{c} K\left(d-1, \frac{c}{2^{d-1}}\right) = \frac{2^{3d}}{c} \cdot \frac{2^{3(d-1)} \cdot 2^{d-1}}{c} K\left(d-2, \frac{c}{2^{(d-1)+(d-2)}}\right),
\]
and so on. There are $d$ factors, each of which is $2^{O(d^2)}/c$, and so $K(d,c) = 2^{O(d^3)}/c^d$.

\smallskip

Let $\gamma = c/2^{d-1}$, and note that $\gamma/p \geq 1$. We decompose $H$ into $d+1$ parts, as follows. Let $H_0 := H$. For $e = 1, \ldots, d$:
\begin{enumerate}
\item Let $B_e$ be the $(d-e)$-uniform hypergraph consisting of all sets $A$ with at least $(\gamma/p)^e$ extensions in $H_{e-1}$.
\item Let $H_e$ be the subhypergraph of $H_{e-1}$ obtained by removing all extensions of hyperedges in $B_e$.
\end{enumerate}
Clearly $\Pr[\hit{H_{e-1}} \neq \emptyset] \leq \Pr[\hit {B_e} \neq \emptyset] + \Pr[\hit {H_e} \neq \emptyset]$. Expanding this, we get
\[
 \Pr[\hit H \neq \emptyset] \leq \sum_{e=1}^d \Pr[\hit{B_e} \neq \emptyset] + \Pr[\hit{H_d} \neq \emptyset].
\]
There are $d+1$ summands on the right-hand side, and one of them needs to be at least $\Pr[\hit{H} \neq \emptyset]/(d+1)$. If $\Pr[\hit{H_d} \neq \emptyset] \geq \Pr[\hit H \neq \emptyset]/(d+1)$ then we choose $H' := H_d$. By construction, $H'$ has branching factor $\gamma/p \leq c/p$, completing the proof in this case.

Suppose next that $\Pr[\hit{B_e} \neq \emptyset] \geq \Pr[\hit H \neq \emptyset]/(d+1)$ for some $e \in \{1,\dots,d\}$. We apply the lemma inductively on $B_e$ and $c := \gamma$ to get a $(d-e)$-uniform hypergraph $B'_e \subseteq B_e$ satisfying
\[ \Pr[\hit{B'_e} \neq \emptyset] \geq \Pr[\hit H \neq \emptyset]/((d+1) K(d-e,\gamma)). \]
We can apply the lemma since $\gamma = c/2^{d-1} \leq c_{\max}(d-e)$ and $p \leq 2^{-(d-1)} \kappa(d-e) c = \kappa(d-e) \gamma$.

For each $I \in B'_e$, we choose a set $N_I$ of $\lfloor (\gamma/p)^e \rfloor \geq (\gamma/p)^e/2$ extensions in $H_{e-1}$. We let $H'$ be the union of the sets $N_I$ for all $I \in B'_e$. Each hyperedge can appear in at most $\binom{d}{d-e} \leq 2^{d-1}$ many $N_I$ (since a hyperedge only has so many subsets of size $d-e$), and so $|H'| \geq (\gamma/p)^e|B'_e|/2^d$.

\smallskip

We start by showing that $H'$ has branching factor $c/p$. We have to show that each set $A$ has at most $(c/p)^{d-|A|}$ extensions in $H'$. We distinguish between two cases, depending on the size of $A$, which we can assume is at most $d-1$ (the condition always holds for larger $A$ for trivial reasons).

Note first that by construction, $H_{e-1}$ has $(e-1)$-bounded branching factor $\gamma/p$, which means that every set $A$ of size at least $d-(e-1)$ has at most $(\gamma/p)^{d-|A|}$ extensions in $H_{e-1}$. Since $H' \subseteq H_{e-1}$, this shows that each set $A$ of size more than $d-e$ has at most $(\gamma/p)^{d-|A|} \leq (c/p)^{d-|A|}$ extensions in $H'$ as well.

Suppose next that $|A| \leq d-e$. For a set of hyperedges $X$, let $A(X)$ be the number of hyperedges in $X$ containing $A$. Each hyperedge in $H'$ containing $A$ belongs to some $N_I$, and so
\[
 A(H') \leq \sum_{I \in B'_e}
 A(N_I)  =
 \sum_{J \subseteq A} \sum_{\substack{I \in B'_e \\ I \cap A = J}} A(N_I).
\]
Fix some particular $J \subseteq A$ and consider all $I$ such that $A\cap I=J$. Since $B'_e$ has branching factor $\gamma/p$, it contains at most $(\gamma/p)^{d-e-|J|}$ many hyperedges $I$ extending $J$.

If $J \subsetneq A$ then each hyperedge counted in $A(N_I)$ contains $I \cup A$, a set of size $|I \cup A| = d-e + |A \setminus J| \geq d-e+1$. Since $N_I \subseteq H_{e-1}$ and $H_{e-1}$ has $(e-1)$-bounded branching factor $\gamma/p$, we see that $A(N_I) \leq (\gamma/p)^{d-(d-e + |A \setminus J|)} = (\gamma/p)^{e-|A\setminus J|}$. In total, the contribution of this $J$ to the sum is at most
\[
 (\gamma/p)^{d-e-|J|} \cdot (\gamma/p)^{e - |A \setminus J|} = (\gamma/p)^{d-|A|}.
\]
If $J = A$ then $A(N_I) \leq (\gamma/p)^e = (\gamma/p)^{e-|A \setminus J|}$ by construction, and so we reach the same conclusion.

Since there are $2^{|A|} \leq 2^{d-e} \leq 2^{d-1}$ summands, we conclude that $A(H') \leq 2^{d-1} (\gamma/p)^{d-|A|} \leq (c/p)^{d-|A|}$, completing the proof that $H'$ has branching factor $c/p$.

\smallskip

Having shown that $H'$ has branching factor $c/p$, we can now appeal to \Cref{cor:bf-hit}, which shows that
\begin{multline*}
 \Pr[\hit{H'} \neq \emptyset] \geq (1 - 2^d c) p^d |H'| \geq
 (1/2) p^d (\gamma/p)^e |B'_e| / 2^d \geq
 (1/2) (\gamma^e/2^d) \Pr[\hit{B'_e} \neq \emptyset] \\ \geq
 \frac{\gamma^e/2^{d+1}}{(d+1) K(d-e, \gamma)} \Pr[\hit H \neq \emptyset] \geq
 \frac{\gamma^e}{2^{3d} K(d-e,\gamma)} \Pr[\hit H \neq \emptyset],
\end{multline*}
using $d+1 \leq 2^d$ and $d \geq 1$.
This completes the proof.
\end{proof}

\subsection{Proof in uniform setting} \label{sec:rub-uniform}

The only two places at which the proof in \Cref{sec:rub-product} accesses the measure directly are \Cref{lem:bf-hit} and \Cref{cor:bf-hit}, so we start by generalizing these two results to the uniform setting: $S \sim \nu_n(k)$, and $p = k/n$. To avoid trivialities, we also assume that $k \geq d$.

\begin{lemma} \label{lem:bf-hit-uniform}
If $c \leq 1$ and $H$ is a $d$-uniform hypergraph with branching factor $c/p$ then for all $I \in H$,
\[
 \Pr_{S\in\nu_n(k)}[\hit H(S) = \{I\} \mid \hit H(S) \ni I] \geq 1 - 2^dc.
\]	
\end{lemma}
\begin{proof}
The original proof uses the bound
\[
 \Pr_{S\sim\mu_p}[S \supseteq J \mid S \supseteq I] \leq p^{|J \setminus I|}.
\]
We claim that this bound holds also for the uniform setting. Indeed, there are $\binom{n-|I|}{k-|I|}$ sets containing $I$, and out of those, $\binom{n-|I \cup J|}{k-|I \cup J|}$ also contain $J$. Therefore the probability is exactly
\[
 \frac{\binom{n-|I \cup J|}{k-|I \cup J|}}{\binom{n-|I|}{k-|I|}} =
 \frac{(k-|I \cup J|+1) \cdots (k-|I|)}{(n-|I \cup J|+1) \cdots (n-|I|)} =
 \frac{k-|I \cup J|+1}{n-|I \cup J|+1} \cdots \frac{k-|I|}{n-|I|}.
\]
Since $(k-\ell)/(n-\ell) \leq k/n$, we see that the probability is at most $p^{(n - |I|) - (n - |I\cup J|)} = p^{|J \setminus I|}$, as needed.

Using this bound, the same statement follows with an identical proof.
\end{proof}

Generalizing \Cref{cor:bf-hit} is messier, since we need a \emph{lower bound} on the event ``$S \supseteq I$'', where $|I| = d$.

\begin{corollary} \label{cor:bf-hit-uniform}
If $c \leq 1$ and $H$ is a $d$-uniform hypergraph with branching factor $c/p$ then for all $I \in H$,
\[
 \Omega(e^{-d}) \cdot (1 - 2^dc) p^d |H| \leq \Pr_{S\in\nu_n(k)}[\hit H \neq \emptyset] \leq p^d |H|.
\]	
\end{corollary}
\begin{proof}
The upper bound follows from the union bound, since the probability of hitting any particular set is at most $p^d$, as shown in the proof of \Cref{lem:bf-hit-uniform}.

The original proof of the lower bound uses $\Pr_{S\sim\mu_p}[S \supseteq I] \geq p^d$ for a set $I$ of size $d$. This inequality no longer holds in the uniform case, instead, we have (for the formula, see the proof of \Cref{lem:bf-hit-uniform})
\[
 \Pr_{S\sim\nu_n(k)}[S \supseteq I] =
 \frac{k-d+1}{n-d+1} \cdots \frac{k}{n} \geq \frac{k(k-1)\cdots(k-d+1)}{k^d} p^d.
\]
We can lower bound the factor in front of $p^d$ by
\[
 \left(1 - \frac{1}{k}\right) \cdots \left(1 - \frac{d-1}{k}\right) \geq
 \left(1 - \frac{1}{d}\right) \cdots \left(1 - \frac{d-1}{d}\right) = \frac{d!}{d^d},
\]
which according to Stirling's approximation is $\Omega(e^{-d})$.
\end{proof}

This adds an additional $\Omega(e^d)$ factor to the very last display of the proof of \Cref{lem:pruning-product-detailed}, but only affects some hidden constants in the statement of the \namecref{lem:pruning-product-detailed}.

{\small
\bibliographystyle{prahladhurl}
\bibliography{DFH-agreement-bib}
}

\end{document}